\documentclass[a4paper, titlepage, oneside, 11pt]{amsart}
\usepackage{amsmath,amssymb,amsfonts,amstext,amscd,array,graphicx,pstricks,xcolor}
\newcommand{\PreserveBackslash}[1]{\let\temp=\\#1\let\\=\temp}
\newcolumntype{C}[1]{>{\PreserveBackslash\centering}p{#1}}

\newtheorem{lem}{Lemma}[section]
\newtheorem{prop}[lem]{Proposition}
\newtheorem{thm}[lem]{Theorem}
\newtheorem{cor}[lem]{Corollary}

\newtheorem{df}[lem]{Definition}
\newtheorem{rem}[lem]{Remark}

\def\R{{\mathbb{R}}}
\def\N{{\mathbb{N}}}

\allowdisplaybreaks
\numberwithin{equation}{section}

\begin{document}

\title[Generalized hydrodynamic limit for the box-ball system]{Generalized hydrodynamic limit\\for the box-ball system}

\author[D.~A.~Croydon]{David A. Croydon}
\address{Research Institute for Mathematical Sciences, Kyoto University, Kyoto 606-8502, Japan}
\email{croydon@kurims.kyoto-u.ac.jp}

\author[M.~Sasada]{Makiko Sasada}
\address{Graduate School of Mathematical Sciences, University of Tokyo, 3-8-1, Komaba, Meguro-ku, Tokyo, 153--8914, Japan}
\email{sasada@ms.u-tokyo.ac.jp}

\begin{abstract}
We deduce a generalized hydrodynamic limit for the box-ball system, which explains how the densities of solitons of different sizes evolve asymptotically under Euler space-time scaling. To describe the limiting soliton flow, we introduce a continuous state-space analogue of the soliton decomposition of Ferrari, Nguyen, Rolla and Wang (cf.\ the original work of Takahashi and Satsuma), namely we relate the densities of solitons of given sizes in space to corresponding densities on a scale of `effective distances', where the dynamics are linear. For smooth initial conditions, we further show that the resulting evolution of the soliton densities in space can alternatively be characterised by a partial differential equation, which naturally links the time-derivatives of the soliton densities and the `effective speeds' of solitons locally.
\end{abstract}


\date{\today}

\maketitle

\section{Introduction}

The box-ball system (BBS) was introduced by Takahashi and Satsuma in \cite{takahashi1990} as a simple example of a discrete soliton system. In particular, it was shown in \cite{takahashi1990} that finite configurations of the BBS could be decomposed into `basic sequences', or solitons, that are conserved by the dynamics. Moreover, it is now understood that by suitably encoding these solitons, it is possible to linearise the dynamics, see the introduction to the inverse scattering method that appears in \cite[Section 3.3]{IKT}, for example. Motivated by studying invariant measures for the BBS supported on bi-infinite configurations, a new soliton decomposition exhibiting linear dynamics was presented by Ferrari, Nguyen, Rolla and Wang in \cite{Ferrari}. In the present work, we explain how the framework of the latter article can be applied to deduce a generalized hydrodynamic limit for the BBS. To describe the limit, we introduce a continuous state-space analogue of the soliton decomposition of \cite{Ferrari}, namely we relate the densities of solitons of given sizes in space to corresponding densities on a scale of `effective distances', where the dynamics are linear. For smooth initial conditions, we further show that the resulting evolution of the soliton densities in space can alternatively be characterised by a partial differential equation, which naturally links the time-derivatives of the soliton densities and the `effective speeds' of solitons locally. We highlight that our results contribute to the growing literature concerning the randomization of the BBS \cite{CKST,CS,CSMont,FG,FG2,Ferrari,Kondo,KL,KLO,Lev,LLPS}, and that it also connects to current work on generalized hydrodynamics \cite{Doyon,DSY,KMP,Sp}.

Let us now be a little more precise about the model of interest here. We will consider a one-sided infinite version of the BBS, with configurations being represented by $\eta=(\eta(x))_{x\in\mathbb{N}}\in\{0,1\}^\mathbb{N}$. We interpret $\eta(x)=1$ as meaning there is a ball in the box at spatial location $x$, and $\eta(x)=0$ as meaning the relevant box is empty. For any such configuration, we can define the dynamics by introducing an auxiliary process $W=(W(x))_{x\in\mathbb{Z}_+}$ taking values in $\mathbb{Z}_+$, which is defined by setting $W_0=0$, and, for $x\geq1$,
\begin{equation}\label{carrierdef}
W(x)=\left\{\begin{array}{ll}
               W(x-1)+1, & \mbox{if }\eta(x)=1,\\
               W(x-1), & \mbox{if }\eta(x)=0\mbox{ and }W(x-1)=0,\\
               W(x-1)-1, & \mbox{if }\eta(x)=0\mbox{ and }W(x-1)>0.
             \end{array}\right.
\end{equation}
(NB.\ We distinguish between $\mathbb{N}:=\{1,2,\dots\}$ and $\mathbb{Z}_+:=\{0,1,\dots\}$.) The quantity $W(x)$ represents the load transported by a `carrier' from $x$ to $x+1$ as it moves from left to right (that is, from negative to positive), picking up each ball as it passes, and dropping off a ball when it is holding at least one ball and sees an empty box. A single pass of the carrier from $0$ to $\infty$ gives one discrete-time step of the BBS dynamics. With this description, we find that the updated configuration $T\eta=(T\eta(x))_{x\in\mathbb{N}}$ is given by
\begin{equation}\label{discretedynamics}
T\eta(x)=\min\left\{1-\eta(x),W(x-1)\right\}.
\end{equation}
Whilst not explicit in the above definition, one readily sees that the basic strings of \cite{takahashi1990}, i.e.\ sequences of the form $(1,0)$, $(1,1,0,0)$, $(1,1,1,0,0,0)$, etc., act like solitons, being preserved by the action of the carrier, travelling at a constant speed (depending on their length) when in isolation, and experiencing interactions when they meet. See Figure \ref{fig:bbs-2soliton} for a simple example of a two soliton interaction in the BBS (adapted from \cite{CKST}).

\newcommand{\BA}{\circle{10}}
\newcommand{\BB}{\circle*{10}}
\newcommand{\BC}{\circle{11}}
\newcommand{\BD}{\circle*{11}}
\makeatother
\begin{figure}[t]
\begin{center}
\noindent
{
\hspace*{-1mm}\BA\BB\BB\BB\BA\BA\BA\BA \BA\BA\BB\BA\BA \BA\BA\BA\BA\BA \BA\BA\BA\BA\BA\BA \BA}\\
\BA\BA\BA\BA\BB\BB\BB\BA \BA\BA\BA\BB\BA \BA\BA\BA\BA\BA \BA\BA\BA\BA\BA\BA \BA\\
\BA\BA\BA\BA\BA\BA\BA\BB \BB\BB\BA\BA\BB \BA\BA\BA\BA\BA \BA\BA\BA\BA\BA\BA \BA\\
\BA\BA\BA\BA\BA\BA\BA\BA \BA\BA\BB\BB\BA \BB\BB\BA\BA\BA \BA\BA\BA\BA\BA\BA\BA \\
\BA\BA\BA\BA\BA\BA\BA\BA \BA\BA\BA\BA\BB \BA\BA\BB\BB\BB \BA\BA\BA\BA\BA\BA \BA\\
\BA\BA\BA\BA\BA\BA\BA\BA \BA\BA\BA\BA\BA \BB\BA\BA\BA\BA \BB\BB\BB\BA\BA\BA \BA\\
\BA\BA\BA\BA\BA\BA\BA\BA \BA\BA\BA\BA\BA \BA\BB\BA\BA\BA \BA\BA\BA\BB\BB\BB\BA
\end{center}
\caption{A two soliton interaction of the box-ball system. (Time runs from the top row to the bottom row.)} \label{fig:bbs-2soliton}
\end{figure}

As was observed in the case of finite configurations in \cite{takahashi1990}, and extended to the case of certain bi-infinite configurations in \cite{Ferrari}, even during interactions, it is possible to identify individual solitons within the configuration of balls. In our setting such a soliton decomposition can be defined on the following subset of configurations:
\begin{equation}\label{omegadef}
\Omega:=\left\{\eta=(\eta(x))_{x\in\mathbb{N}}\in\{0,1\}^\mathbb{N}:\:\limsup_{x\rightarrow\infty}\sum_{x'=1}^x\left(1-2\eta(x')\right)=\infty\right\};
\end{equation}
details of the procedure for doing so will be presented in Section \ref{soldecompsec}. Related to this, we introduce the notation, for $i,x\in\mathbb{N}$, $k\in\mathbb{Z}_+$,
\begin{equation}\label{sigmaidef}
\sigma_i(x,k):=\mathbf{1}_{\left\{\exists\mbox{ a soliton of size $i$ in $T^k\eta$ starting at spatial location $x$}\right\}},
\end{equation}
where $T^k\eta$ is the configuration after $k$ steps of the dynamics. Throughout, we will only ever consider initial configurations $\eta\in\Omega$ that incorporate solitons of size at most $I$ for some $I\in\mathbb{N}$, and so we also define
\begin{equation}\label{omegaidef}
\Omega_I:=\left\{\eta\in\Omega:\:\sigma_i(x,0)=0,\:\forall i>I,\:x\in\mathbb{N}\right\}.
\end{equation}
For random initial configurations with distribution supported on this set, our first main result (see Theorem \ref{reallythemainthm} below) is a generalized hydrodynamic limit for the rescaled empirical distributions of solitons of different sizes. In particular, we consider for $i\in\{1,2,\dots,I\}$ and $N\in\mathbb{N}$,
\[\pi^{N,t}_i(du):=\frac{1}{N}\sum_{x \in \N} \sigma_i\left( x, \lfloor Nt \rfloor\right) \delta_{x/N}(du),\qquad u,t\in\mathbb{R}_+,\]
where $\delta_x$ is the probability measure on $\mathbb{R}_+$ placing all of its mass at $x$.

Before proceeding, it will be helpful to discuss the effect of interactions on soliton speeds. In the case of a single size $i$ soliton overtaking a size $j$ soliton, where $i>j$, it is an easy exercise to check that the larger soliton receives a push forward of $2j$ when compared to its free evolution, and the smaller soliton a push backward by the same amount. (See Figure \ref{fig:bbs-2soliton} again.) The entirety of such interactions was considered in \cite{Ferrari} in the case of random, spatially stationary, bi-infinite configurations, with a system of equations being presented that relates the asymptotic speeds of solitons of different sizes, dependent on the densities of the variously-sized solitons. The latter formulae will also be relevant for determining the effective speeds of solitons in our limiting model. In particular, for $\rho \in \mathbb{R}_+^I:=[0,\infty)^I$ satisfying $\sum_{i=1}^I 2i \rho_i <1$, where $\rho_i$ represents the density of solitons of size $i$, we define the effective speeds $v^{\mathrm{eff}}(\rho)=(v^{\mathrm{eff}}_i(\rho))_{i=1}^I$ via the equations:
\begin{equation}\label{effectivespeed}
v^{\mathrm{eff}}_i(\rho)=v_i-\sum_{j=1}^I 2(i\wedge j) \rho_j(v^{\mathrm{eff}}_j(\rho)-v^{\mathrm{eff}}_i(\rho)),\qquad i\in\{1,2,\dots,I\},
\end{equation}
where $v_i:=i$ is the velocity of a soliton of size $i$ in isolation. (NB.\ Throughout the article, we use the notation $i\wedge j:=\min\{i,j\}$ and $i\vee j:=\max\{i,j\}$.) To give an explicit formula for $v^{\mathrm{eff}}(\rho)$, let $M(\rho)$ be the matrix defined by
\begin{equation}\label{matrix}
\left\{
  \begin{array}{ll}
    \displaystyle{M_{ii}(\rho)=1-\sum_{j \neq i} 2(i \wedge j) \rho_j}, &\vspace{3pt} \\
    \displaystyle{M_{ij}(\rho)=2(i\wedge j)\rho_j}, & i\neq j.
  \end{array}
\right.
\end{equation}
It is possible to check that $M(\rho)$ is invertible (see Lemma \ref{minvert} below), and so $v^{\mathrm{eff}}(\rho)$ is uniquely given by
\begin{equation}\label{effv}
v^{\mathrm{eff}}(\rho)=M(\rho)^{-1}v,
\end{equation}
where $v^{\mathrm{eff}}(\rho)$ and $v$ are column vectors with entries $v^{\mathrm{eff}}_i(\rho)$ and $v_i$ respectively.

As the final ingredient we need to state Theorem \ref{reallythemainthm}, we introduce the set of initial density profiles that the result will cover. To this end, we let $\mathcal{D}_{\mathrm{density}}$ be the collection of $(\rho_i)_{i=1}^I\in \mathcal{C}^1(\mathbb{R}_+,\mathbb{R}_+)^I$, where $\mathcal{C}^1(\mathbb{R}_+,\mathbb{R}_+)$ will be defined shortly, that satisfy both
\begin{equation}\label{cond1}
\sum_{i=1}^I2i\rho_i(u)<1,\qquad \forall u\in\mathbb{R}_+,
\end{equation}
and
\begin{equation}\label{cond2}
\sum_{i=1}^I i \sup_{u}\frac{\rho_i(u)}{1-\sum_{j=1}^I2(i\wedge j)\rho_j(u)}<\frac{1}{2}.
\end{equation}
As will be expanded upon later, the first condition here enables us to reparameterize solitons from their spatial position into the scale of effective distances that we will introduce, and the second condition is a density condition for solitons in their effective distances that is preserved by the dynamics. The requirement $(\rho_i)_{i=1}^I\in \mathcal{C}^1(\mathbb{R}_+,\mathbb{R}_+)^I$ is a differentiability assumption that ensures we have enough smoothness to solve the partial differential equation of interest. In particular, for any $q \in \N$, we write $\mathcal{C}^q(\R_+,\R_+)$ for the set of functions $f:\R_+\rightarrow\R_+$ such that the derivatives of $f^\mathbb{R}$ of order less than or equal to $q$ are continuous everywhere in $\R$, where $f^{\R}(u):=0$ for $u < 0$, and $f^{\R}(u):=f(u)$ for $u \ge 0$. We highlight that $f \in \mathcal{C}^q(\R_+,\R_+)$ implies $f(0)=0$, and so $\mathcal{C}^1(\R_+,\R_+) \subsetneq C^1(\R_+,\R_+)$. We also denote by $C_0(\mathbb{R}_+,\mathbb{R})$ the set of $f:\R_+\rightarrow\R$ that are continuous on $\mathbb{R}_+$ (right-continuous at 0), and compactly supported.

\begin{thm}\label{reallythemainthm} Fix $I\in \mathbb{N}$ and $\rho^0=(\rho^0_i)_{i=1}^I\in \mathcal{D}_{\mathrm{density}}$. For each $N\in\mathbb{N}$, consider the BBS initiated from a random configuration $\eta$ with distribution $\mathbf{P}_N$ supported on $\Omega_I$, and suppose that, for every $\varepsilon>0$ and $(F_i)_{i=1}^I \in C_0(\mathbb{R}_+,\mathbb{R})^I$,
\[\lim_{N\rightarrow\infty} \mathbf{P}_N\left(\sup_{i\in\{1,2,\dots,I\}} \left| \int_{\R_+} F_i(u) \pi_i^{N,0}(du) - \int_{\R_+} F_i(u) \rho^0_i(u)du \right|  > \varepsilon\right) =0.\]
It then holds that, for every $\varepsilon>0$, $t\in(0,\infty)$ and $(F_i)_{i=1}^I \in C_0(\mathbb{R}_+,\mathbb{R})^I$,
\[\lim_{N\rightarrow\infty} \mathbf{P}_N\left(\sup_{i\in\{1,2,\dots,I\}} \left| \int_{\R_+} F_i(u)\pi_i^{N,t} (du)-  \int_{\R_+} F_i (u)\rho_i(u,t)du \right|  > \varepsilon\right) =0,\]
where $(\rho_i(u,t))_{u,t\in\mathbb{R}_+,i=1,2,\dots,I}$ is the unique classical solution of the partial differential equation
\begin{equation}\label{PDErho}
\left\{\begin{array}{l}
\vspace{3pt}\displaystyle{\partial_t \rho_i =- \partial_u \left( v_i^{\mathrm{eff}}(\rho) \rho_i\right),}\\
\displaystyle{\rho_i(\cdot,0)=\rho^0_i(\cdot),}
\end{array}
\qquad i=1,2,\dots,I,\right.
\end{equation}
amongst the class of functions $\rho\in C^1(\mathbb{R}_+^2,\mathbb{R}_+)^I$ satisfying $\rho (\cdot,t)  \in \mathcal{D}_{\mathrm{density}}$ for all $t \ge 0$.
\end{thm}
\medskip

\begin{rem}\label{3rem}
(a) Whilst we restrict to the case of finite $I$ in Theorem \ref{reallythemainthm}, we anticipate that essentially the same strategy can be applied to deduce the corresponding result when $I=\infty$. Since this involves a number of additional technical arguments, however, we prefer to postpone this to a future work.\\
(b) Using the time-reversibility of the BBS (i.e.\ that $T^{-1}$ is given by running the carrier from right-to-left, see \cite{CKST, takahashi1990}), it would be possible to adapt the proof of Theorem \ref{reallythemainthm} to deduce the corresponding result for the BBS started from initial configurations $\eta=(\eta(x))_{x\in\mathbb{Z}}$ such that $\eta(x)=0$ for $x\geq 0$ and for which $(\eta(-x))_{x\in\mathbb{N}}$ satisfies the conditions of the theorem (including that $I<\infty$, which ensures a finite maximum soliton speed).\\
(c) Another natural extension would be to consider truly bi-infinite initial configurations, namely $\eta=(\eta(x))_{x\in\mathbb{Z}}$ such that $\eta(x)=1$ infinitely often as $x\rightarrow\pm\infty$ (cf.\ \cite{CKST,Ferrari}). As in \cite[Theorem 2.1]{Ferrari}, in this situation one would need to be careful about the handling of solitons spanning 0 in the discrete and continuous soliton decompositions (and their dynamics). We leave the treatment of this issue as an open problem.
\end{rem}

The previous theorem is in fact a corollary of two further results, which we now set out. The first of these is a generalized hydrodynamic limit for the integrated densities of solitons. This picture allows for a wider class of initial conditions, including densities that are step functions (see Figure \ref{solfig} for an illustrative example). Moreover, it gives a clear description of how the continuous system evolves in terms of the linear evolution of the solitons seen on their effective scale.

The discrete object of interest is now defined by setting
\begin{equation}\label{psiindef}
\psi_i^N(u,t):=\pi_i^{N,t}\left([0,u]\right)=\frac{1}{N}\sum_{x=1}^{\lfloor Nu \rfloor}\sigma_i\left(x,\lfloor Nt\rfloor\right).
\end{equation}
For each fixed time, the limiting continuous analogue will be an element of
\begin{equation}\label{ddef}
\mathcal{D}:=\left\{(\psi_i)_{i=1}^I\in\mathcal{C}^I:\phi_I\in \mathcal{C}^{\uparrow},
\sum_{i=1}^I i \sup_{u_1,u_2}\frac{\psi_i(u_1)-\psi_i(u_2)}{\phi_i(u_1)-\phi_i(u_2)}<\frac{1}{2}\right\},
\end{equation}
where
\begin{align}
\mathcal{C}&:=\left\{f\in C(\R_+,\R_+):\:f(0)=0,\:f\mbox{ non-decreasing}\right\},\label{cdef}\\
\mathcal{C}^{\uparrow}&:=\left\{f\in \mathcal{C}:\:f\mbox{ strictly increasing}, \lim_{u\rightarrow\infty}f(u)=\infty\right\},\label{cupdef}
\end{align}
and for $(\psi_i)_{i=1}^I\in\mathcal{C}^I$, we define
\begin{equation}\label{phiidef}
\phi_i(u):=u-\sum_{j=1}^I2(i\wedge j)\psi_j(u),\qquad i=1,2,\dots,I.
\end{equation}
In particular, $\psi_i$ represents the integrated density of size $i$ solitons in space, and, as will be discussed further below, $\phi_i(u)$ is the effective distance accumulated by size $i$ solitons over the spatial interval $[0,u]$. The condition $\phi_I\in\mathcal{C}^\uparrow$ (cf.\ \eqref{cond1}) ensures that the changes of scale are well-defined, see Lemma \ref{philem}. Moreover, the final condition in the definition of $\mathcal{D}$ (cf.\ \eqref{cond2}) represents an assumption on the soliton density in terms of the effective distance, which will be shown to be preserved by the dynamics, see Proposition \ref{tpprop}. We next define $\Upsilon:\mathcal{D}\rightarrow \Upsilon(\mathcal{D})\subseteq\mathcal{C}^I$ by setting $\Upsilon(\psi):=\bar{\psi}$, where $\bar{\psi}=(\bar{\psi}_i)_{i=1}^I$ is given by
\begin{equation}\label{barpsidef}
\bar{\psi}_i:=\psi_i\circ\phi_i^{-1};
\end{equation}
we note that, as can be understood from Proposition \ref{gobet} below, $\bar{\psi}_i$ represents the integrated density of size $i$ solitons with respect to their effective distance in the continuous state-space model. Importantly, the effective scaling, or scattering, map $\Upsilon:\mathcal{D}\rightarrow\Upsilon(\mathcal{D})$ is a bijection, and so we can also recover the spatial picture from the one in terms of effective distances, see Proposition \ref{tpprop}.

We highlight that, in addition to the generalized hydrodynamic limit results that we present, a main contribution of this article is the introduction of the picture we have just set out involving the functions $(\psi_i)_{i=1}^I$ and $(\bar{\psi}_i)_{i=1}^I$, as well as the explicit description we give of the scattering map $\Upsilon$ and its inverse. Indeed, whilst the discrete versions of $\psi_i$ and $\bar{\psi}_i$ are essentially contained in \cite{Ferrari}, as is a description of the link between these, since \cite{Ferrari} only deals with the asymptotically homogeneous case, the continuous limits of all the maps in question are linear, and their understanding reduces to a computation of the relevant coefficients. (Similarly, as already noted, the asymptotic effective speeds of solitons of different sizes were also constant by assumption in \cite{Ferrari}.) We believe our more conceptual approach clarifies the soliton decomposition in a way that will be useful in other contexts (including for other discrete integrable systems) where asymptotic inhomogeneity is present. We further note that the choice to base our definitions on the soliton decomposition in the discrete case on \cite{Ferrari} is with the prospective of eventually studying the two-sided case (see Remark \ref{3rem}(c)), to which the decomposition in \cite{Ferrari} already applies. Since we are only dealing with the one-sided case here, however, it would alternatively been possible to present the definitions in terms of the more classical linearization of the BBS, based on the Kerov-Kirillov-Reshetikhin (KKR) bijection (\cite[Section 3]{IKT}, see also \cite{KNTW}).

As the last piece of notation we need to state our generalized hydrodynamic limit for integrated densities, for $\bar{\psi}\in\mathcal{C}^I$ we define $\theta_t\circ\bar{\psi}\in\mathcal{C}^I$ by setting
\begin{equation}\label{thetatdef}
(\theta_t\circ\bar{\psi})_i(z):=\bar{\psi}_i\left((z-v_it)\vee 0\right),
\end{equation}
where we recall $v_i:=i$;
this operator gives the free evolution of soliton densities, that is, the temporal dynamics of soliton densities with respect to their effective distances. (This is by definition in the continuous case, see \eqref{thetacon} for the appearance of $\theta_t$ in the scaling limit of the discrete model.)

\begin{thm}\label{mainthm} Fix $I\in \mathbb{N}$ and $\psi^0=(\psi^0_i)_{i=1}^{I}\in\mathcal{D}$. For each $N\in\mathbb{N}$, consider the BBS initiated from a random configuration $\eta$ with distribution $\mathbf{P}_N$ supported on $\Omega_I$, and suppose that, for every $\varepsilon>0$ and $u_0\in(0,\infty)$,
\begin{equation}\label{initconv}
\lim_{N\rightarrow\infty}\mathbf{P}_N\left(\sup_{i\in\{1,2,\dots,I\}}\sup_{u\leq u_0}\left|\psi_i^N(u,0)-\psi_i^0(u)\right|>\varepsilon\right)=0.
\end{equation}
It then holds that, for every $\varepsilon>0$, $t\in(0,\infty)$ and $u_0\in(0,\infty)$,
\[\lim_{N\rightarrow\infty}\mathbf{P}_N\left(\sup_{i\in\{1,2,\dots,I\}}\sup_{u\leq u_0}\left|\psi_i^N(u,t)-\psi_i(u,t)\right|>\varepsilon\right)=0,\]
where
\begin{equation}\label{BBSflow}
\psi_i(u,t):=\left(\Upsilon^{-1}\circ\theta_t\circ\Upsilon\circ\psi^0\right)_i(u).
\end{equation}
\end{thm}
\medskip

\begin{figure}
\begin{center}
\includegraphics[width=0.47\textwidth,height=0.35\textwidth]{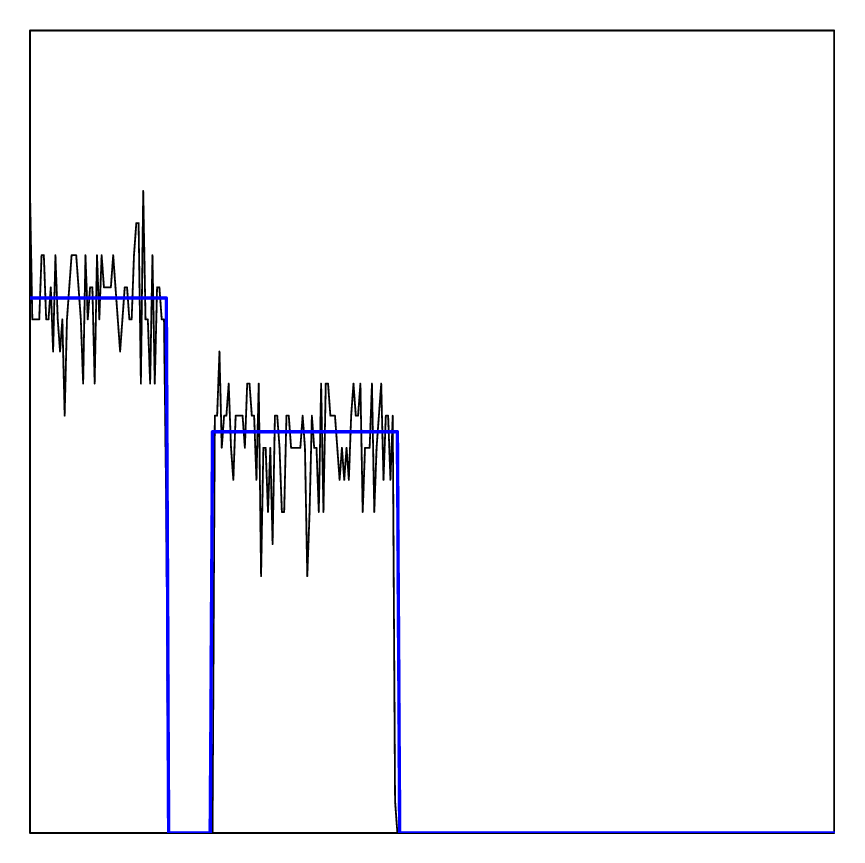}
\includegraphics[width=0.47\textwidth,height=0.35\textwidth]{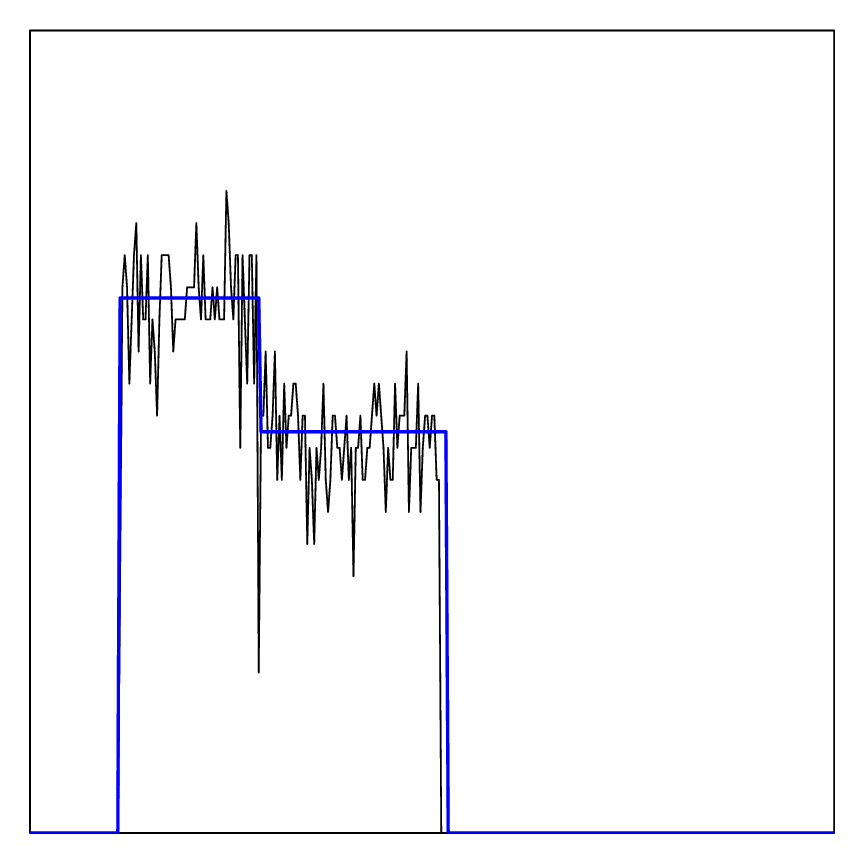}
\includegraphics[width=0.47\textwidth,height=0.35\textwidth]{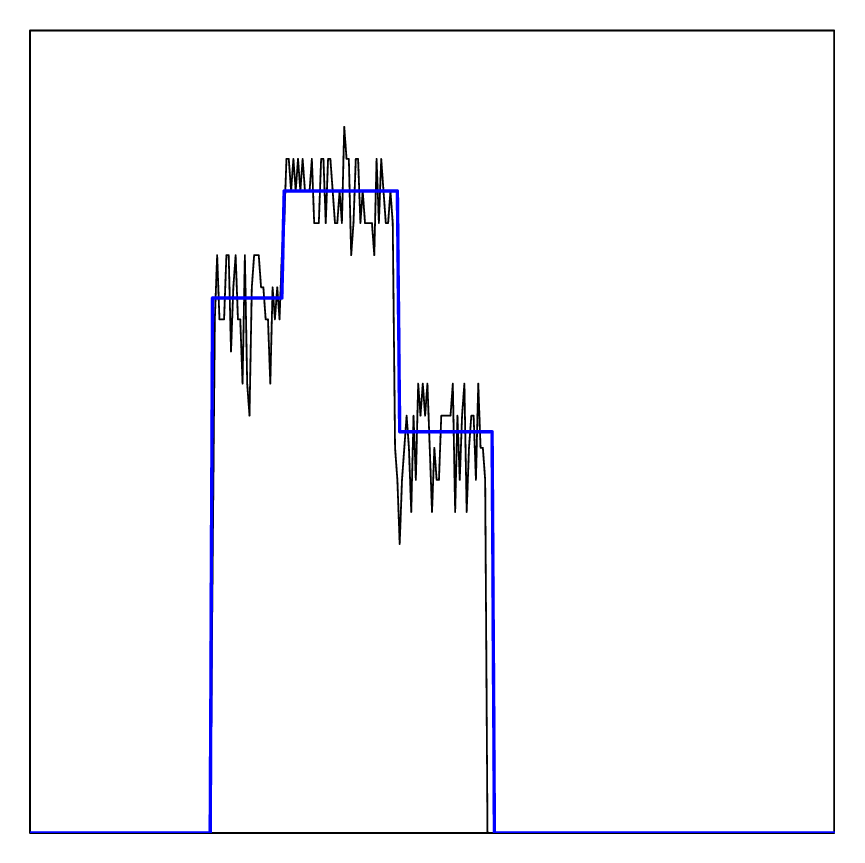}
\includegraphics[width=0.47\textwidth,height=0.35\textwidth]{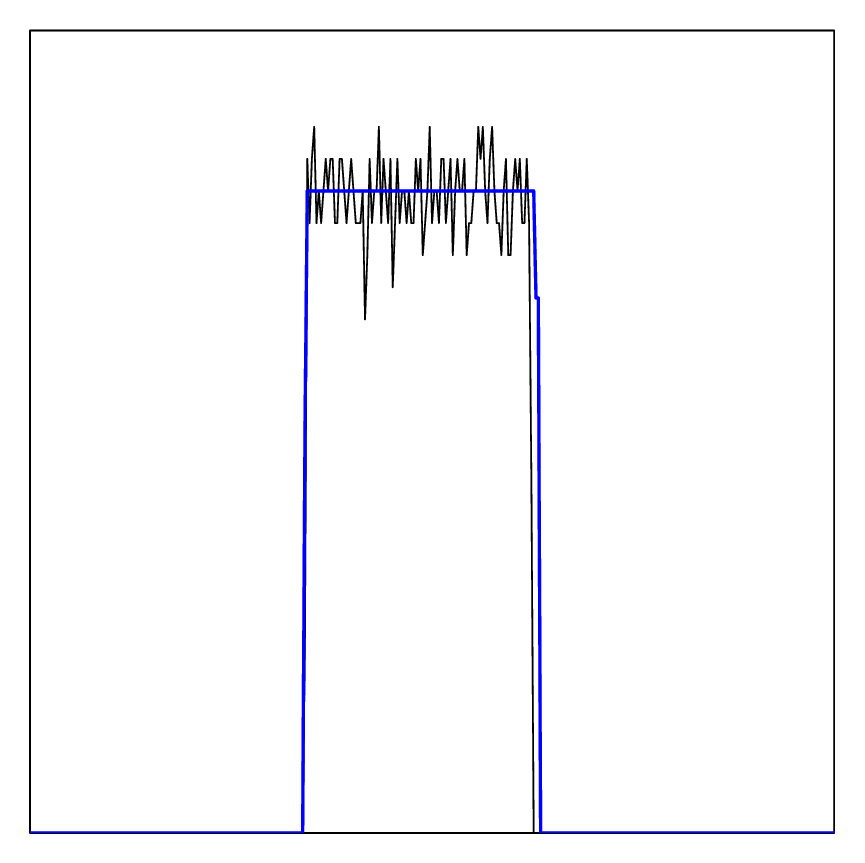}
\includegraphics[width=0.47\textwidth,height=0.35\textwidth]{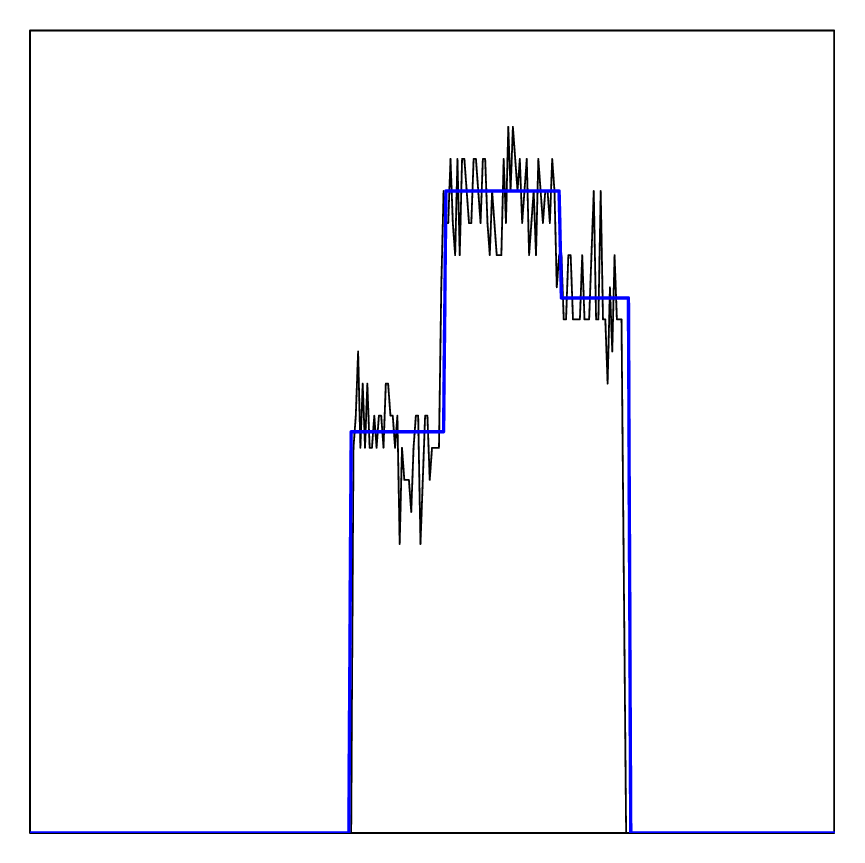}
\includegraphics[width=0.47\textwidth,height=0.35\textwidth]{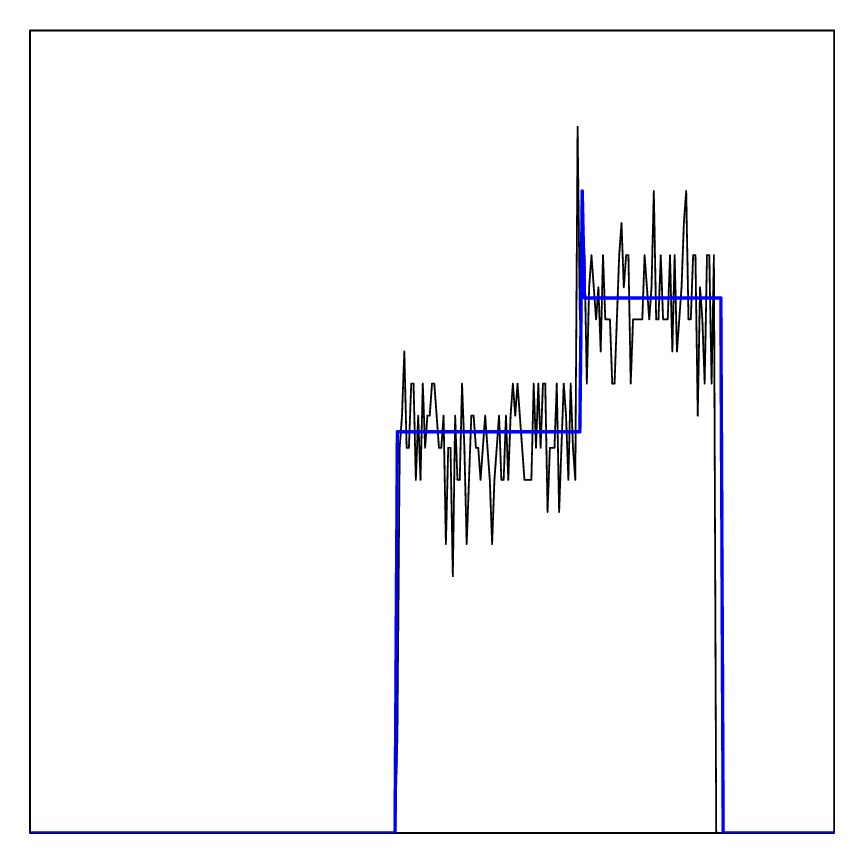}
\includegraphics[width=0.47\textwidth,height=0.35\textwidth]{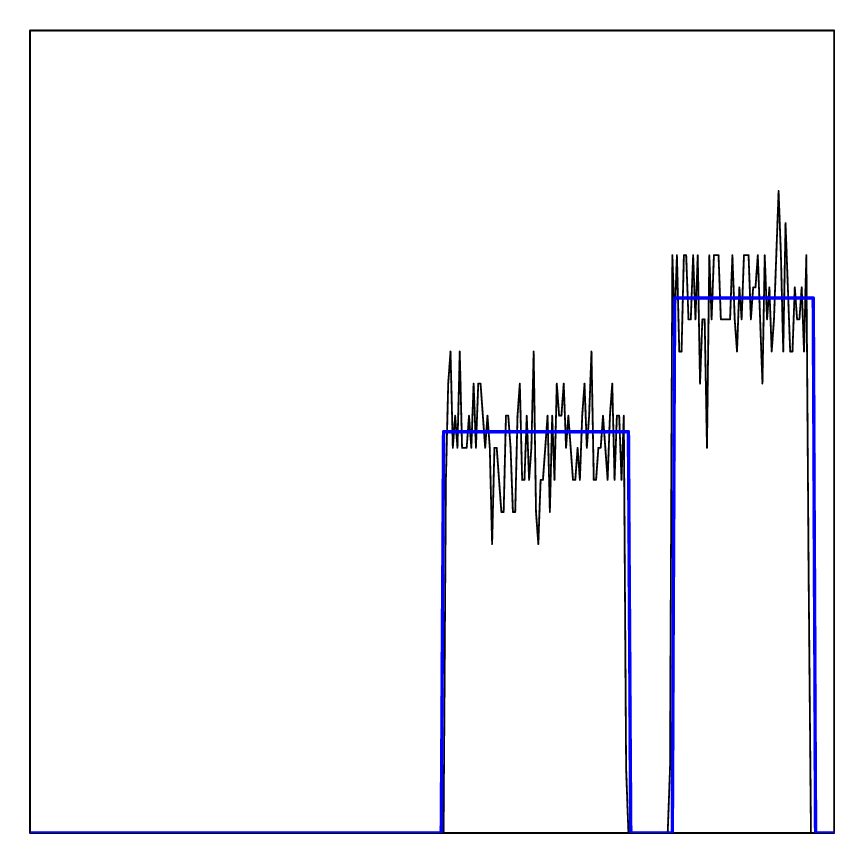}
\includegraphics[width=0.47\textwidth,height=0.35\textwidth]{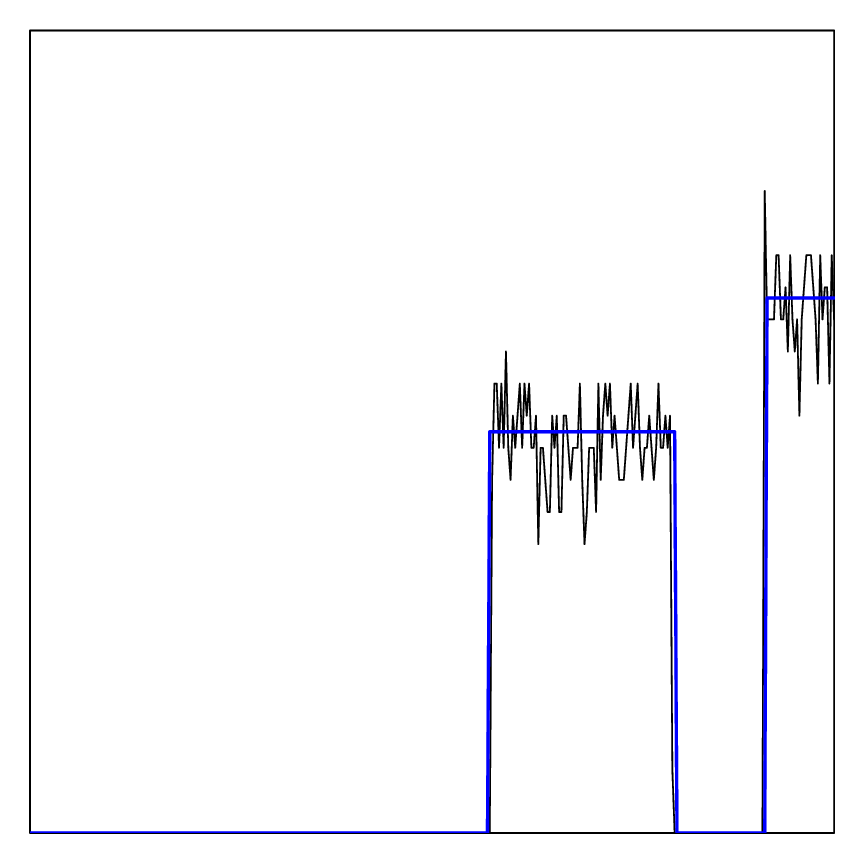}
\end{center}
\caption{A region of size $2$ solitons overtaking a region of size $1$ solitons (with the initial condition being randomly generated from a Bernoulli measure on the `slot decomposition' of the configuration, see Subsection \ref{soldecompsec}). In each picture, the black line shows the particle density locally (each data point is the average of 50 sites in $\mathbb{N}$), and the blue line shows the generalized hydrodynamic limit of the system, as given by Theorem \ref{mainthm}.}\label{solfig}
\end{figure}

A simple corollary of the preceding two theorems is the following limiting result for the integrated particle density and empirical particle distribution. For its statement, we define $\eta(x,k):=T^k\eta(x)$. See Figure \ref{pdefig} for an example.

\begin{cor}\label{maincor1} Under the conditions of Theorem \ref{mainthm}, it holds that, for every $\varepsilon>0$, $t\in(0,\infty)$ and $u_0\in(0,\infty)$,
\[\lim_{N\rightarrow\infty}\mathbf{P}_N\left(\sup_{u\leq u_0}\left|\frac{1}{N}\sum_{x=1}^{\lfloor Nu\rfloor }\eta\left(x,\lfloor Nt\rfloor\right)-\sum_{i=1}^Ii\psi_i(u,t)\right|>\varepsilon\right)=0.\]
Moreover, under the conditions of Theorem \ref{reallythemainthm}, it further holds that, for every $\varepsilon>0$, $t\in(0,\infty)$ and $F \in C_0(\mathbb{R}_+,\mathbb{R})$,
\[\lim_{N\rightarrow\infty} \mathbf{P}_N\left(\left| \frac{1}{N} \sum_{x \in \N} F \left(\frac{x}{N}\right)\eta(x, \lfloor Nt \rfloor) -  \int_{\R_+} F (u)  \rho_{\mathrm{particle}}(u,t)du \right|  > \varepsilon\right) =0,\]
where $\rho_{\mathrm{particle}}(u,t):= \sum_{i=1}^I  i \rho_i(u,t)$ satisfies
\begin{equation*}
\partial_t \rho_{\mathrm{particle}}=- \partial_u \left( \sum_{i=1}^I i v_i^{\mathrm{eff}}(\rho) \rho_i\right).
\end{equation*}
\end{cor}
\medskip

As the second input into Theorem \ref{reallythemainthm}, we have that, under additional regularity assumptions to Theorem \ref{mainthm}, the limiting dynamics at \eqref{BBSflow} are given by a partial differential equation. We present our result in this direction on two levels, both for the integrated densities of solitons, and for the densities themselves. See Figure \ref{pdefig} for an example solution to the density-level partial differential equation.

\begin{thm}\label{mainthm2}
(a) Let $\psi^0=(\psi^0_i)_{i=1}^{I}\in\mathcal{D}^{(1)}:=\mathcal{D}\cap \mathcal{C}^1(\mathbb{R}_+,\mathbb{R}_+)^I$,
and define $\psi_i(u,t)$ by \eqref{BBSflow}. It is then the case that $\psi_i(\cdot,\cdot) \in C^1(\mathbb{R}_+^2,\mathbb{R}_+)$ for each $i$, $\psi(\cdot,t) \in \mathcal{D}^{(1)}$ for all $t \ge0$, and
\begin{equation}\label{PDEpsi}
\begin{cases}
&\vspace{3pt} \displaystyle{\partial_t\psi_i=-v_i^{\mathrm{eff}}\left(\partial_u\psi\right)\partial_u\psi_i}, \qquad i=1,2,\dots,I,  \\
& \displaystyle{\psi(\cdot,0)=\psi^0(\cdot).}
\end{cases}
\end{equation}
Moreover, $(\psi_i(u,t))_{u,t\in\mathbb{R}_+,i=1,2,\dots,I}$ is the unique solution of \eqref{PDEpsi} in the class of functions $\psi \in C^1(\mathbb{R}_+^2,\mathbb{R}_+)^I$ satisfying $\psi(\cdot,t) \in \mathcal{D}^{(1)}$ for all $t \ge0$.\\
(b) Let $\rho^0=(\rho^0_i)_{i=1}^{I} \in \mathcal{D}_{\mathrm{density}}$, and define $\psi_i(u,t)$ by \eqref{BBSflow} with initial condition given by
\[\psi^0_i(u):=\int_0^u\rho^0_i(u')du',\qquad i=1,2,\dots,I.\]
It is then the case that $\psi_i(\cdot,\cdot) \in C^2(\mathbb{R}_+^2,\mathbb{R}_+)$ for each $i$, $\rho_i(u,t):=\partial_u\psi_i(u,t)$ satisfies $\rho_i \in C^1(\mathbb{R}_+^2,\mathbb{R}_+)$ for each $i$, $\rho (\cdot,t) \in \mathcal{D}_{\mathrm{density}}$
for all $t \ge0$, and \eqref{PDErho} holds. Moreover, $(\rho_i(u,t))_{u,t\in\mathbb{R}_+,i=1,2,\dots,I}$ is the unique solution of \eqref{PDErho} in the class of functions $\rho \in C^1(\mathbb{R}_+^2,\mathbb{R}_+)^{I}$ satisfying $\rho(\cdot,t) \in \mathcal{D}_{\mathrm{density}}$ for all $t \ge0$.
\end{thm}
\medskip

\begin{figure}
\begin{center}
\includegraphics[width=0.47\textwidth,height=0.36\textwidth]{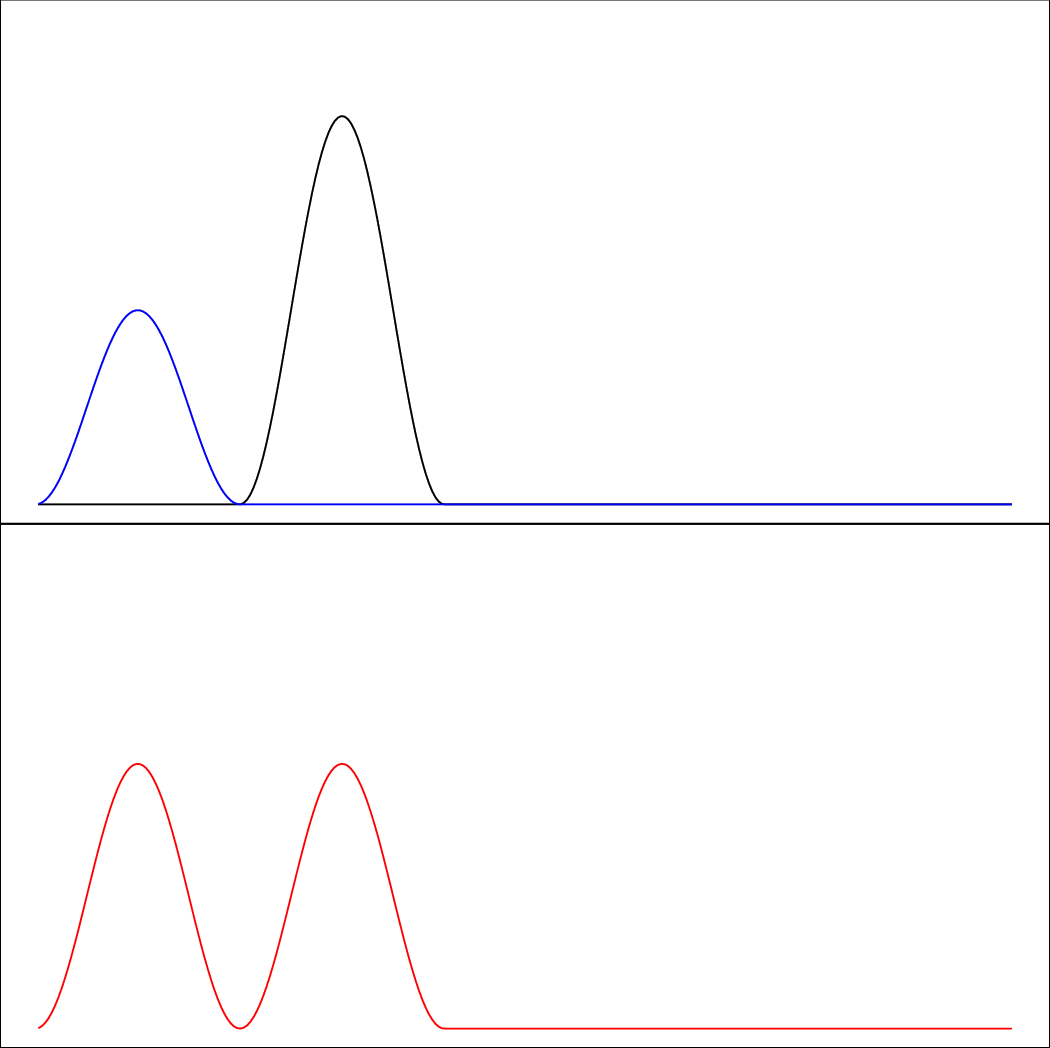}
\includegraphics[width=0.47\textwidth,height=0.36\textwidth]{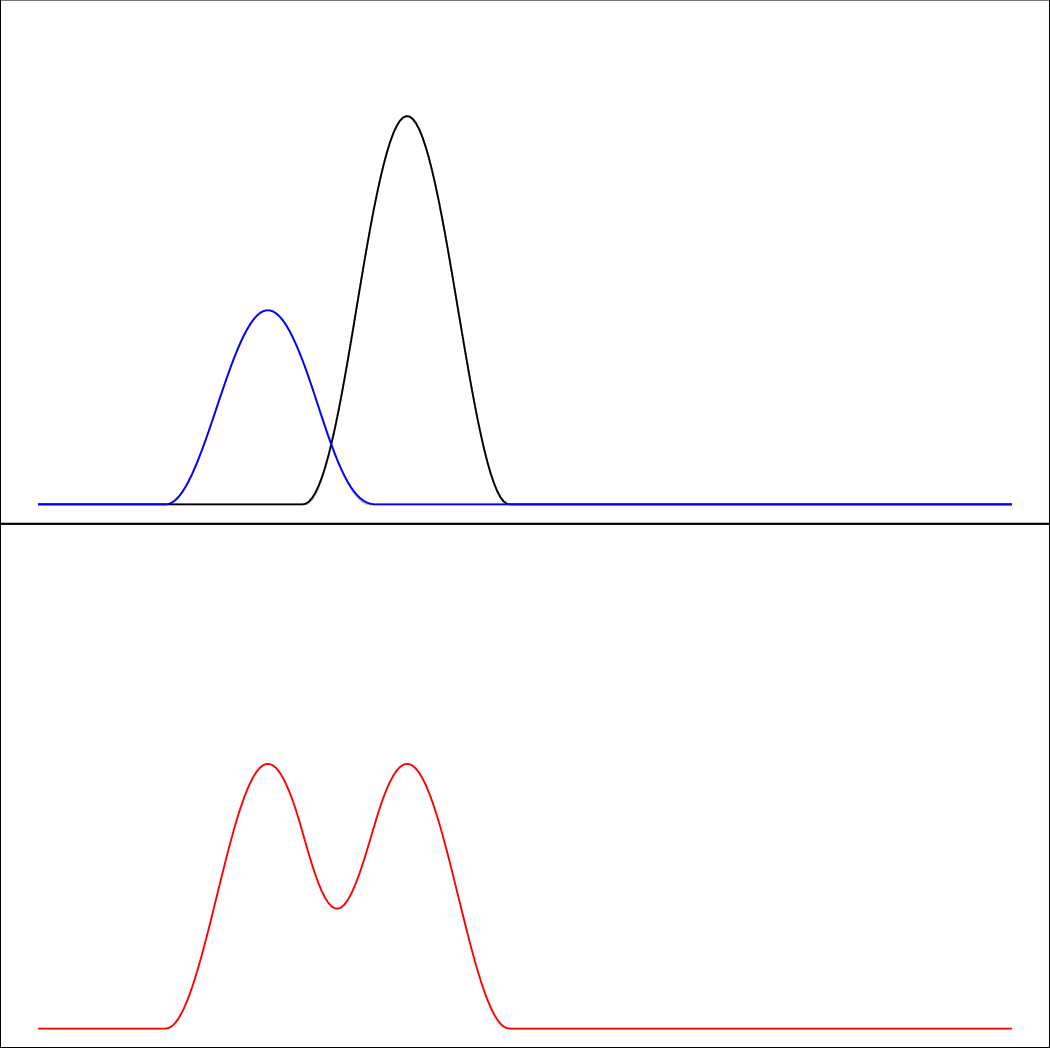}
\smallskip

\includegraphics[width=0.47\textwidth,height=0.36\textwidth]{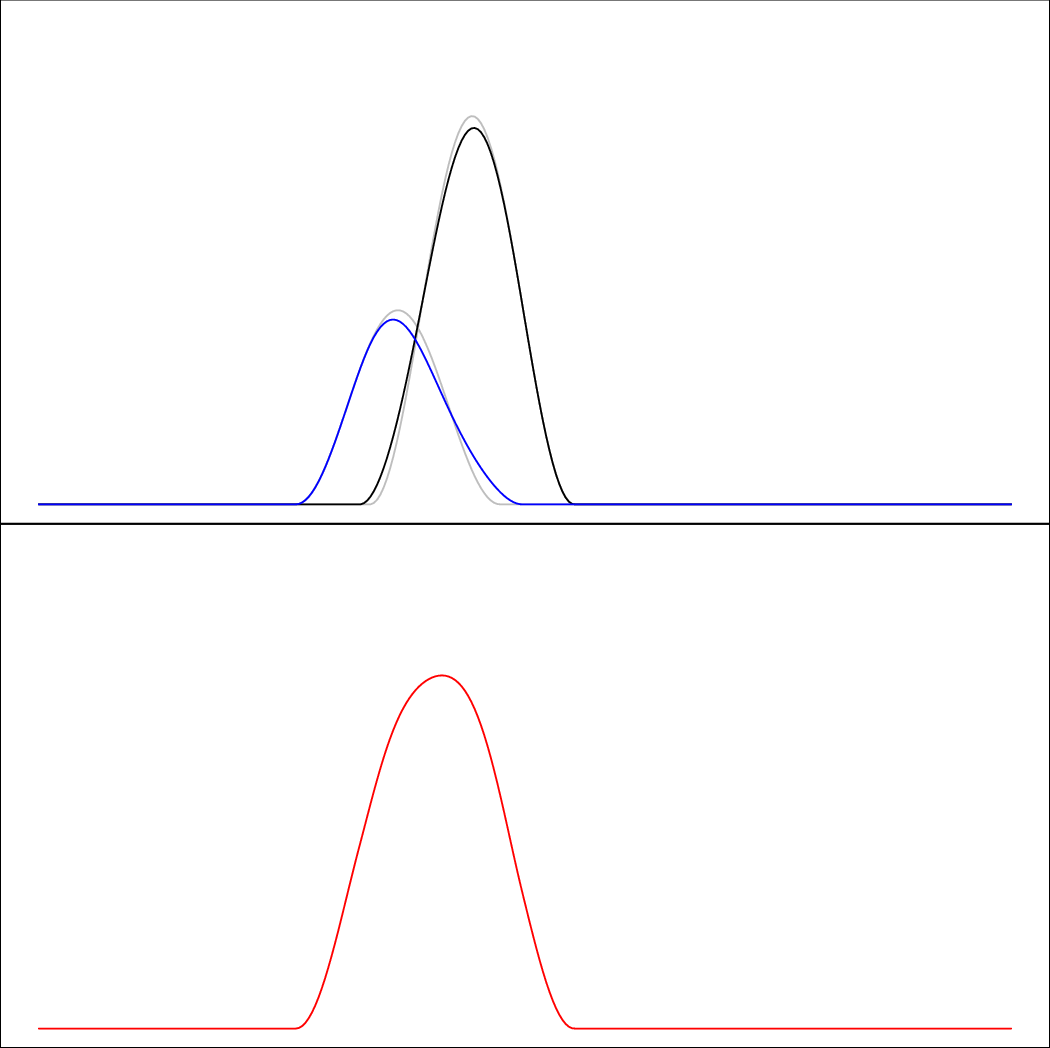}
\includegraphics[width=0.47\textwidth,height=0.36\textwidth]{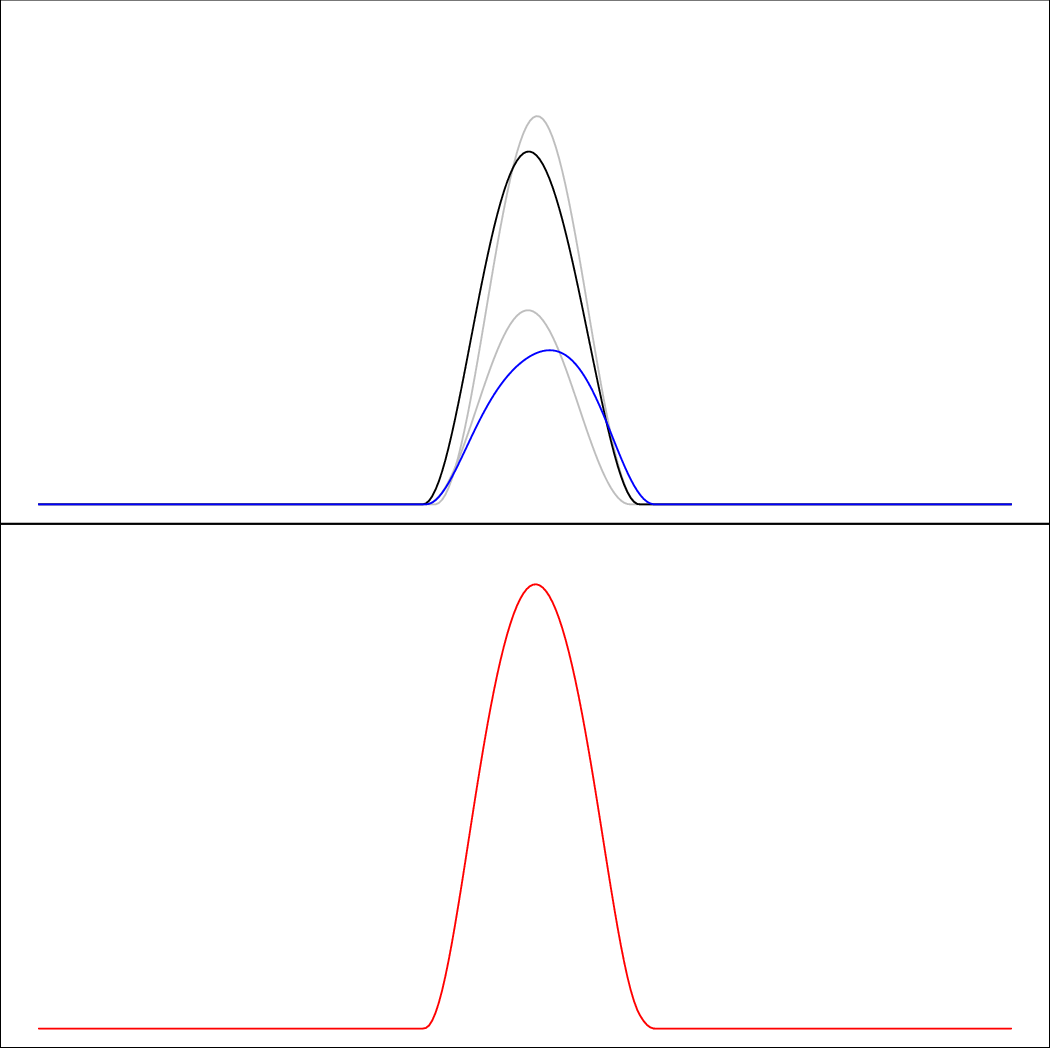}
\smallskip

\includegraphics[width=0.47\textwidth,height=0.36\textwidth]{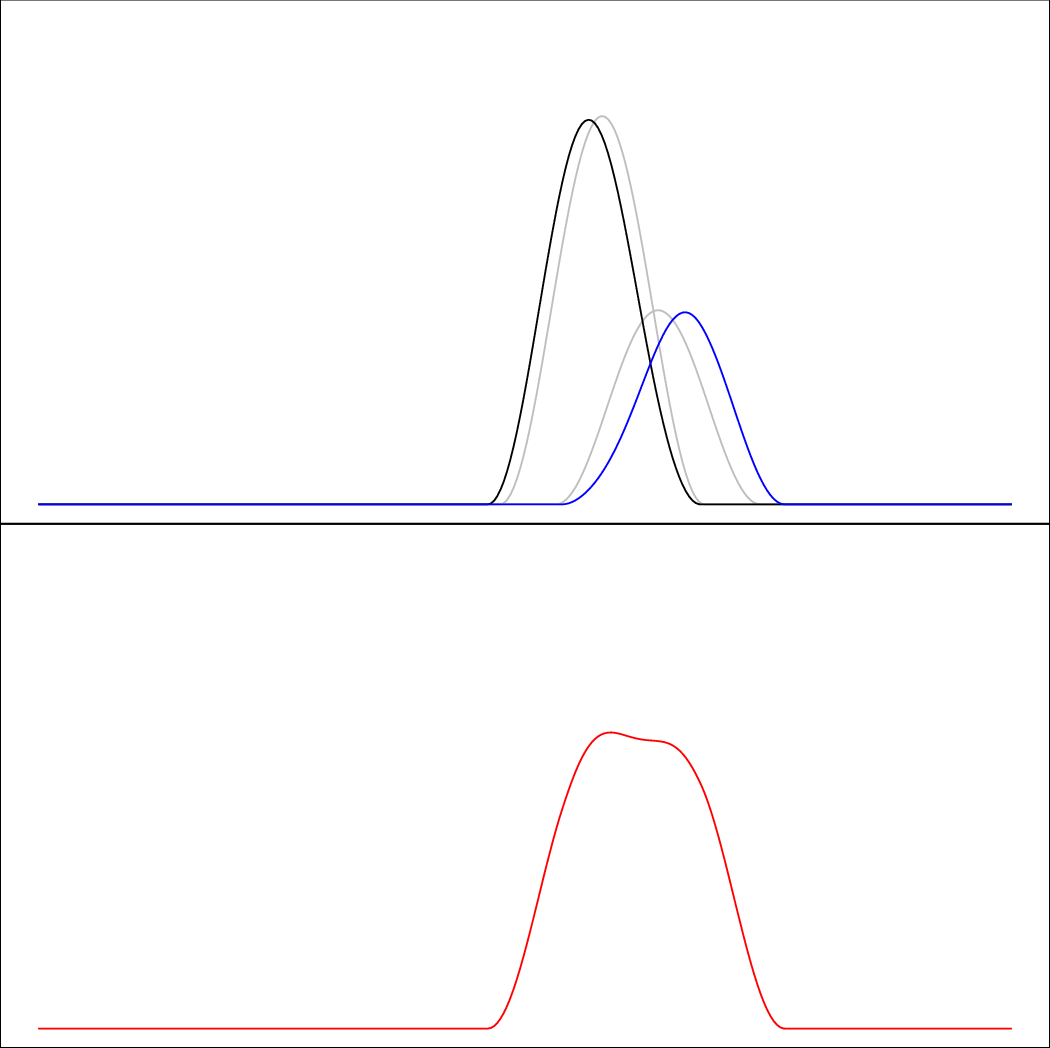}
\includegraphics[width=0.47\textwidth,height=0.36\textwidth]{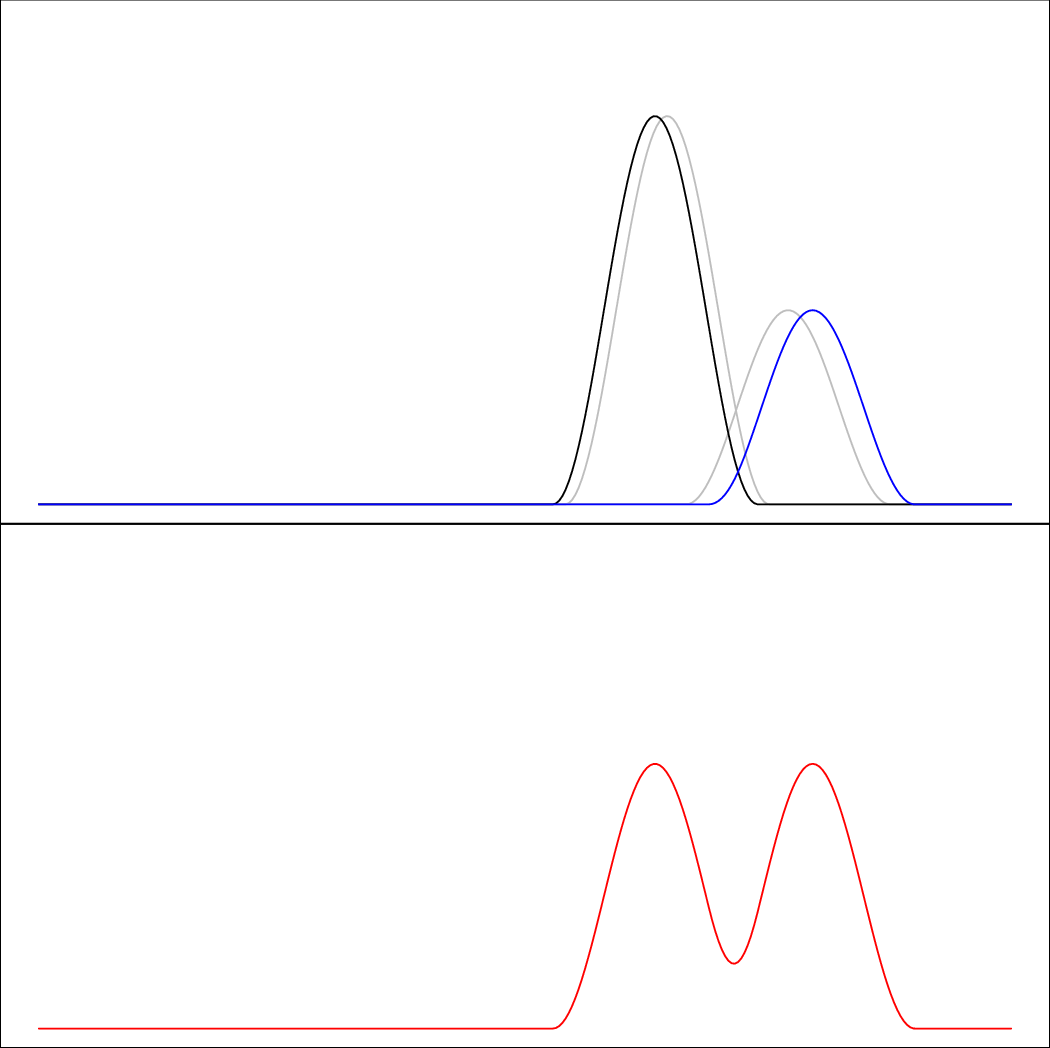}
\smallskip

\includegraphics[width=0.47\textwidth,height=0.36\textwidth]{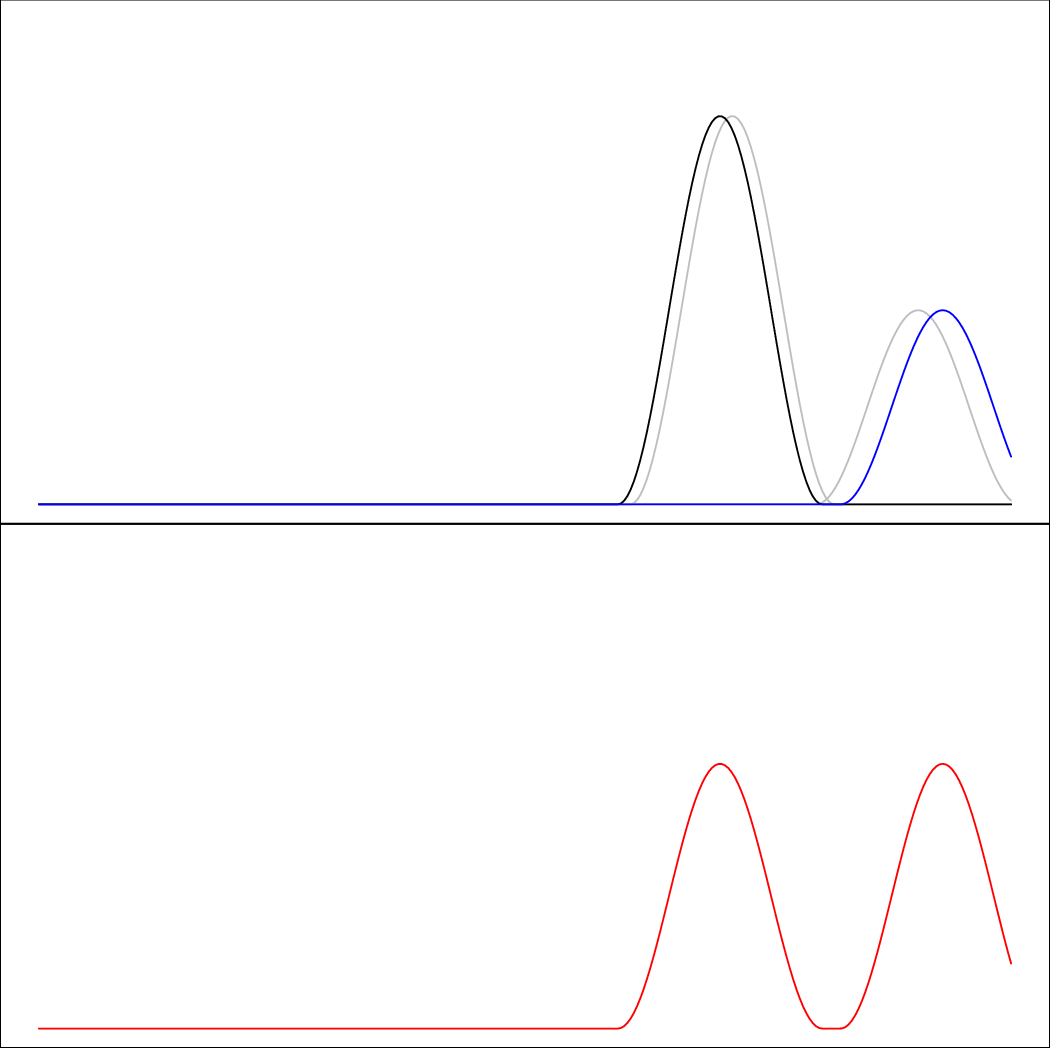}
\includegraphics[width=0.47\textwidth,height=0.36\textwidth]{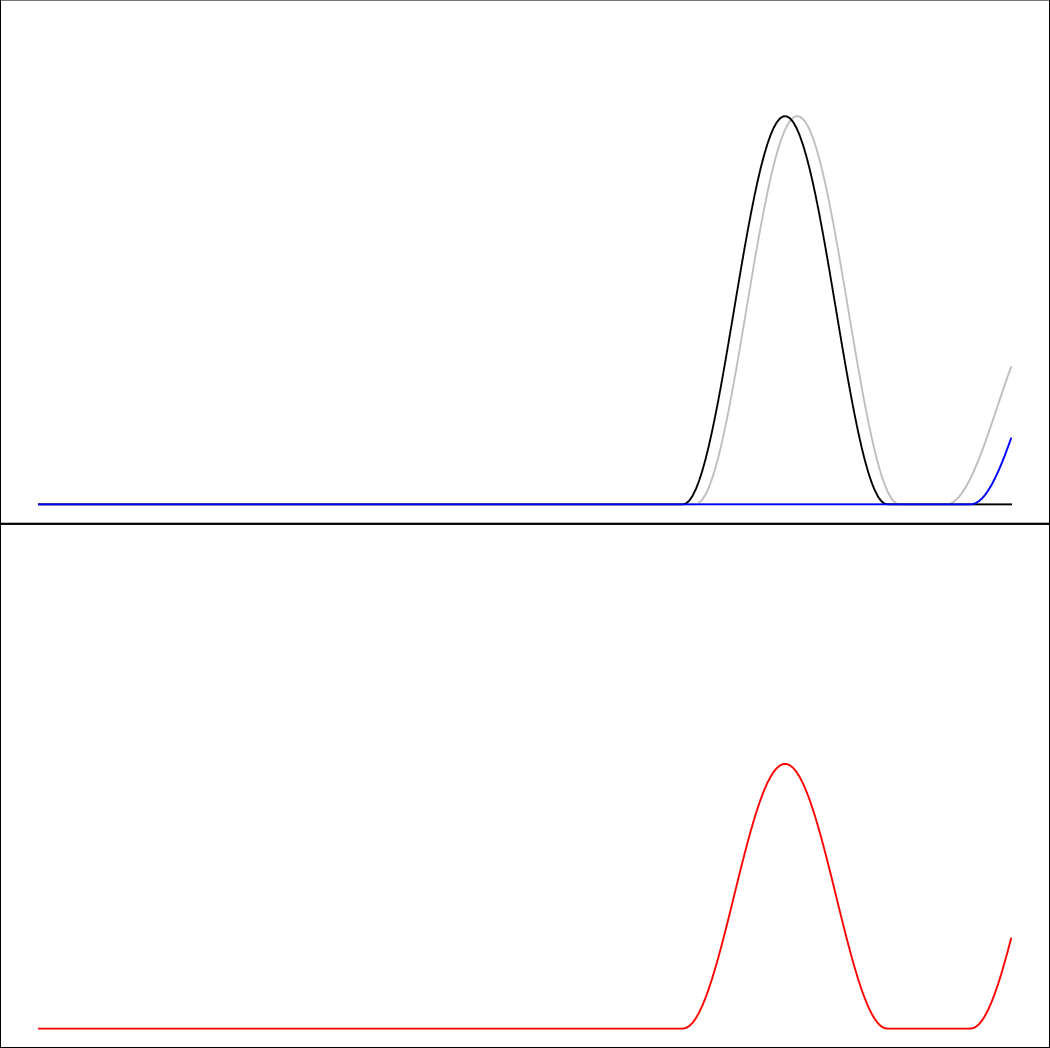}
\smallskip
\end{center}
\caption{The black curve shows $\rho_1$, the blue curve shows $\rho_2$, with dynamics given by Theorem \ref{mainthm2}. The grey curves show how the soliton densities would evolve without any interaction. By comparing with these, the deformation during the interaction and resulting phase shift are clear. The red curve shows the particle density, as studied in Corollary \ref{maincor1}.}\label{pdefig}
\end{figure}

To place our results into context, we continue by presenting some background on the dynamical study of soliton distributions, and more generally the hydrodynamics of integrable systems. A central example of an integrable system having solitary wave solutions (solitons) is the Korteweg-de Vries (KdV) equation. (Note that the BBS of interest in this article can in fact be derived as an ultra-discretization of the discrete KdV equation, which is in turn a certain discretization of KdV equation \cite{TH}.) The study of the evolution of KdV soliton distribution goes back a long way. Indeed, in 1971, Zakharov \cite{Zakharov} derived a kinetic equation describing the spatial evolution of the soliton distribution in a rarefied gas of the KdV solitons,
\[\partial_t f_s =- \partial_u (\tilde{v}_s(f) f_s),\]
where $f_s=f_s(u,t)$ is the density at space-time point $(u,t)$ of solitons with respect to their spectral parameter $s$. Here, $\tilde{v}_s(f)$ represents the effective velocity of a single `trial' soliton with spectral parameter $s$ in a given soliton distribution $f=(f_{s})$ under a rarefied gas condition, and is explicitly expressed as
\[\tilde{v}_s(f) = 4s^2 - \frac{1}{s} \int_0^{\infty} \log \left| \frac{s+r}{s-r} \right|f_r (4r^2 -4s^2 )dr,\]
where $4s^2$ is the speed of a KdV soliton with parameter $s$ in `empty' space. Later, El \cite{El} generalized this result to the case of a dense gas and derived the kinetic equation
\[\partial_t f_s =- \partial_u ( v^{\mathrm{eff}}_s(f) f_s),\]
where the transport velocity $v^{\mathrm{eff}}_s(f)$ now satisfies
\[v^{\mathrm{eff}}_s(f) = 4s^2 - \frac{1}{s} \int_0^{\infty} \log \left| \frac{s+r}{s-r} \right| f_r \left(v^{\mathrm{eff}}_r(f) -v^{\mathrm{eff}}_s(f) \right) dr.\]
We highlight that this does not give an explicit expression for the velocity, but is merely a linear integral equation. This kinetic equation was extended to other classical integrable systems with solitary waves in \cite{EK}, and generalized to a nonlinear integro-differential system
\begin{equation*}
\begin{cases}
& \vspace{3pt}\partial_t f_s =- \partial_u \left( v^{\mathrm{eff}}_s(f) f_s\right),\\
& \displaystyle{v^{\mathrm{eff}}_s(f) = v_s - \frac{1}{s} \int_0^{\infty} \kappa(s,r) f_r \left(v^{\mathrm{eff}}_r(f) -v^{\mathrm{eff}}_s(f) \right)  dr},
\end{cases}
\end{equation*}
with explicitly given model-dependent parameters $(v_s)$ and symmetric interaction kernel $\kappa(s,r)$, see \cite{EKPZ}. In the latter work, to study this type of system, El et al.\ introduced its finite-component `cold-gas' hydrodynamic reductions. In particular, these reductions are obtained via the delta-function ansatz for the density $f_s=\sum_{i=1}^I f_i \delta (s-s_i)$, where $0<s_1<s_2<\dots <s_I$ are arbitrary numbers. The resulting system of hydrodynamic conservation laws is given by
\begin{equation}\label{jio1}
\partial_t \rho_{i} =- \partial_u ( v^{\mathrm{eff}}_{i}(\rho) \rho_{i}), \qquad i=1,2,\dots, I,
\end{equation}
where the densities $\rho_{i} := s_if_i$ and the velocities $ v^{\mathrm{eff}}_{i}(\rho)$ are related by
\begin{equation}\label{jio2}
v^{\mathrm{eff}}_i(\rho) = v_{s_i} -\sum_{j=1}^I \kappa_{ij}  \rho_j \left(v^{\mathrm{eff}}_j (\rho) -v^{\mathrm{eff}}_i(\rho) \right),
\end{equation}
with $\kappa_{ij}:=\frac{1}{s_i s_j}\kappa(s_i,s_j)$. The form of the partial differential equation and the equation of the effective velocity are exactly the same as \eqref{PDErho} and \eqref{effectivespeed}, respectively, with $v_{s_i}=i$ and $\kappa_{ij}=2(i \wedge j)$ for the BBS. In \cite{EKPZ}, various properties of the system described by \eqref{jio1} and \eqref{jio2} are presented, including its linear degeneracy, as well as formulae for velocities and some exact solutions with particular symmetries.

In recent years, a similar nonlinear integro-differential system has also been derived from quantum integrable systems by applying generalized hydrodynamics (GHD) \cite{Doyon,Sp}. The theory of GHD is an extension of hydrodynamics to integrable systems, constructed on generalized Gibbs ensembles instead of Gibbs ensembles. It explains that the Euler-scale evolution of quasi-particles in integrable systems, which are associated to solitons or more generally conserved quantities, is described by the GHD equations:
\begin{equation*}
\begin{cases}
&\vspace{3pt} \partial_t \rho_p =- \partial_u \left( v^{\mathrm{eff}}_p(\rho) \rho_p\right), \\
& \displaystyle{v^{\mathrm{eff}}_p(\rho) = E'(p) + \int  K(p,p') \rho_p(p') \left(v^{\mathrm{eff}}_{p'}(\rho) -v^{\mathrm{eff}}_p(\rho) \right)  dp'},
\end{cases}
\end{equation*}
where $\rho_p$ is the density of quasi-particles of `bare' momentum $p$, and $E'(p)$ is the group velocity of them, with $E(p)$ being the energy function. The kernel $K(p,p')$ is called a `scattering' or `two-body' shift, and is determined by examining the two-body scattering of asymptotic excitations for quasi-particles with momentum $p$ and $p'$, see \cite[Section 3]{Doyon}. Recent developments reveal that the theory is applicable to a wide family of integrable systems, including classical and quantum gases, chains and field theory models. Our result shows that the GHD equations are also relevant for the BBS. That the theory of GHD is suitable for cellular automata (of which the BBS is one) has been conjectured (see \cite[Introduction]{Doyon}), but not been shown rigorously hitherto. Moreover, the strategy of our proof follows exactly that which might be expected for a rigorous derivation of GHD given the interpretation in \cite[Section 4.3]{Doyon}, that is, ``GHD is the fluid equation obtained by applying the inverse of the scattering map to the Liouville equations [that describe the fluid motion of free particles]''. We also mention work by Kuniba et al.\ that studies the explicit solution of the GHD equation for the BBS starting from a step initial function of particle densities, where the initial soliton densities are those associated with the Bernoulli product measure of given particle density \cite{KMP}. (NB.\ For this, it is necessary to consider an infinite number of soliton sizes, that is, take $I=\infty$.) In the latter work, it is confirmed that there is a series of density plateaux emerging from the initial condition, and explicit expressions for the height, speed and position of the plateaux are given.

Finally, we give a remark on an essential difference between typical hydrodynamic limits for interacting stochastic systems and our results. In the study of hydrodynamic limits for interacting stochastic systems, the characterization of all equilibrium states is an important step. In particular, in hydrodynamics, macroscopic properties of the system can be deduced from the profile of macroscopic parameters. For example, for the symmetric simple exclusion process, the equilibrium states are completely characterized by the density of particles, and under an appropriate condition on the microscopic initial measures, the convergence of the empirical measure associated with any local function (i.e.\ one depending upon the configuration in a finite neighbourhood of the origin) holds, namely,
\[\frac{1}{N}\sum_{x} f\circ\tau_x \left(\eta^{tN^2}\right)\delta_{x/N}(du) \to \tilde{f}(\rho(t,u))du\]
as $N\rightarrow\infty$, where $f$ is an arbitrary local function, $\tau_x$ is a spatial shift by $x$, $\eta^t$ is the configuration at time $t$, $\tilde{f}(\rho)$ is the expected value of $f$ under the equilibrium measure with density $\rho$, and $\rho(t,u)$ is the solution of the hydrodynamic limit equation for the density of particles \cite{KLbook}. On the other hand, since there is no mixing effect in our dynamics, the soliton densities do not characterize all invariant measures. In fact, as shown in \cite{Ferrari}, for each given sequence of soliton densities $(\rho_i)_{i}$, there are infinitely many invariant measures for the BBS. In particular, the macroscopic profile of soliton densities does not contain enough information to characterize all macroscopic properties, and so the convergence of the empirical measure associated with an arbitrary local function does not necessarily hold. On the contrary, we do not know which measure should be used to compute $\tilde{f}$. However, if we assume our initial condition to be `a local equilibrium' in a particular sense, such as the elements of the `slot decomposition' $(\zeta_i(m))_{i,m}$ (as introduced in Subsection \ref{soldecompsec} below) are independent, then such a convergence statement may also hold for the BBS.

The remainder of the article is organised as follows. In Section \ref{discsec}, we recall the soliton decomposition from \cite{Ferrari}, and explain how the dynamics are linear with respect to this. We then present our continuous analogue in Section \ref{contsec}. Section \ref{hdlsec} contains a proof of the generalized hydrodynamic limit for the integrated densities of solitons that is stated as Theorem \ref{mainthm}. In Section \ref{msec}, we check that for densities that are not too large, the matrix defined at \eqref{matrix} is invertible, and therefore the effective speeds of solitons are well-defined. Finally, in Section \ref{pdesec}, we establish the partial differential equation descriptions of Theorem \ref{mainthm2}, which enables us to complete the proofs of Theorem \ref{reallythemainthm} and Corollary \ref{maincor1}.

\section{Discrete soliton decomposition and dynamics}\label{discsec}

In this section, we recall the soliton decomposition of \cite{Ferrari}, and explain how linear dynamics for solitons in this frame of reference induce the evolution of the BBS, as defined at \eqref{discretedynamics}.

\subsection{Soliton decomposition of particle configuration}\label{soldecompsec}

Before getting to the details of the soliton decomposition, it will be convenient to introduce the path encoding for BBS configurations of \cite{CKST}. In particular, for a configuration $\eta=(\eta(x))_{x\in\mathbb{N}}\in\{0,1\}^\mathbb{N}$, let $(S(x))_{x\in\mathbb{Z}_+}$ be the nearest-neighbour path on the integers given by $S(0):=0$, and
\[S(x)-S(x-1):=1-2\eta(x),\qquad \forall x\in\mathbb{N}.\]
If we define the past maximum of the path $(M(x))_{x\in\mathbb{Z}_+}$ by setting
\[M(x):=\max_{x'\leq x}S(x'),\qquad \forall x\in\mathbb{Z}_+,\]
then we have that the carrier process $(W(x))_{x\in\mathbb{Z}_+}$, as introduced at \eqref{carrierdef}, is given by
\begin{equation}\label{wpath}
W(x)=M(x)-S(x),\qquad \forall x\in\mathbb{Z}_+,
\end{equation}
see \cite[Lemma 2.1]{CKST}. In the soliton decomposition, we will consider the `excursions' of $W$ between the `records' of $S$, where this terminology is defined as follows.
\begin{itemize}
  \item A spatial location $x\in\mathbb{Z}_+$ is a record of $S$, or simply a record, if and only if $x$ is equal to 0 or $x\geq 1$ and $M(x)>M(x-1)$ (that is, $x$ is equal to $0$ or a new maximum for $S$). Equivalently, from \eqref{wpath}, we see that non-zero records are precisely those $x\in\mathbb{N}$ for which $W(x)=W(x-1)=0$.
  \item The excursion of $W$ between a pair of records $x_0$ and $x_1$ is the path segment $(W(x))_{x=x_0}^{x_1-1}$ that starts and ends at 0, makes jumps of $+1$ or $-1$, and remains non-negative throughout. We remark that the excursions of $W$ are separated by its `flat' segments, which, as already observed, correspond to non-zero records. Moreover, it can be the case that $x_0=x_1-1$, in which case the relevant excursion has zero length.
\end{itemize}
Note that for $\eta\in\Omega$, as defined at \eqref{omegadef}, we have that $M(x)\rightarrow\infty$ as $x\rightarrow\infty$, and so $S$ admits an infinite number of records, and the excursions of $W$ between these are all of finite length. Figure \ref{patheg} shows an example path of $S$ up to the first non-zero record, and the corresponding $M$ and $W$.

\begin{figure}
\begin{center}
\includegraphics[width=0.6\textwidth]{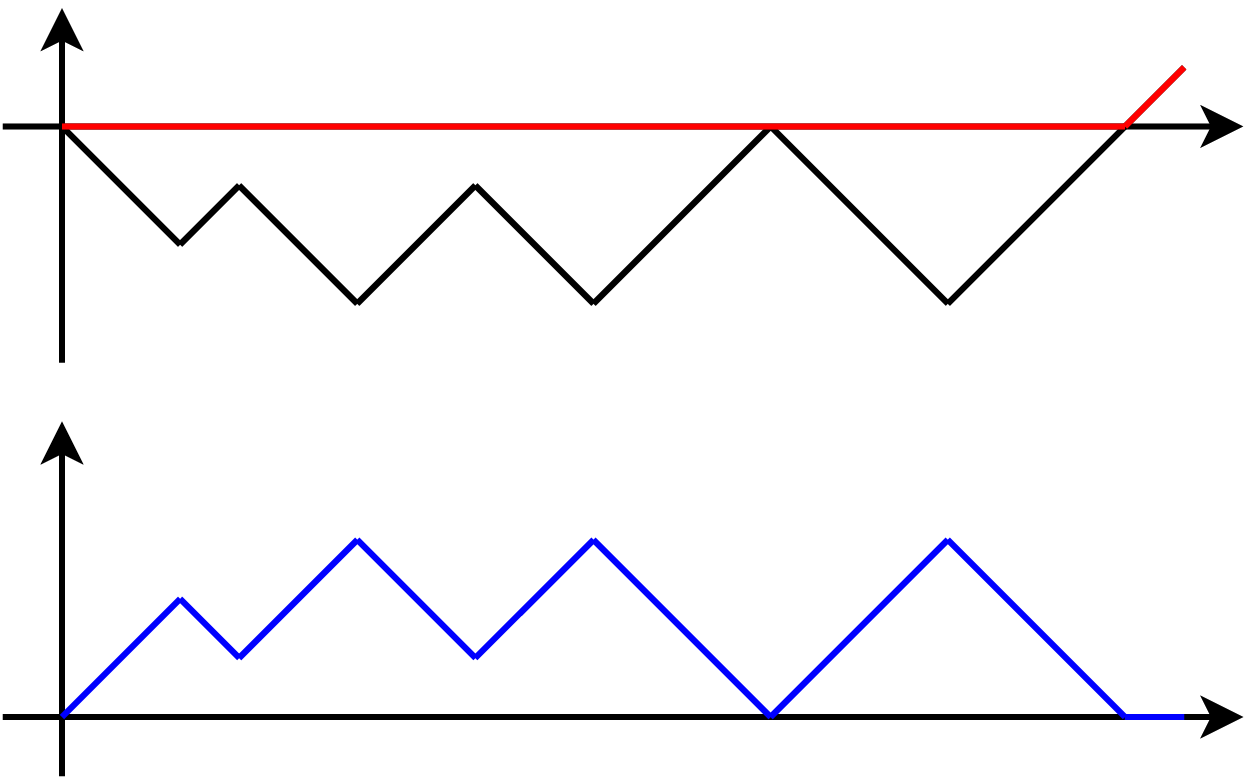}
\end{center}
\caption{Path encoding $S$ (top graph, black), past maximum $M$ (top graph, red) and carrier process (bottom graph, blue), shown up to the first non-zero record.}\label{patheg}
\end{figure}

We are now ready to describe an algorithm that, for any $\eta\in\Omega$, identifies every element of $\mathbb{Z}_+$ as being a record, or an element of a soliton of a finite size. As a result, we will be able to give a precise definition of $\sigma_i(x)=\sigma_i(x,0)$, which was formally defined at \eqref{sigmaidef}. Since the algorithm operates independently on each of the excursions of $W$, we will simply describe it on the part of the configuration $(\eta(x))_{x=x_0+1}^{x_1-1}$ between a pair of records $x_0$ and $x_1$ such that $S(x_1)=S(x_0)+1$. As an aid to the reader, we present a worked example in Figure \ref{solident} (corresponding to the configuration used to produce Figure \ref{patheg}). In particular, whilst the part of the configuration we are considering remains non-empty, we do the following (cf.\ \cite{Ferrari,takahashi1990}).
\begin{itemize}
\item Let $i$ be the length of the left-most run of consecutive $0$s (or $1$s) in the configuration that is followed by a run of consecutive $1$s (or $0$s, respectively) of at least the same length. Group the elements of this run with the first $i$ elements of the subsequent run. The $2i$ grouped elements are identified as a size $i$ soliton.
\item Remove the identified soliton, and repeat until the configuration is empty.
\end{itemize}
Since there are an equal number of $0$s and $1$s in the original part of the configuration, this algorithm will indeed terminate in an empty configuration, meaning that each of the $x\in\{x_0+1,x_0+2,\dots,x_1-1\}$ will have been included in a finite size soliton. Moreover, we note that the spatial locations of the elements of a soliton of size $i$ are given by a set of lattice sites $\{x'_1,x'_2,\dots,x'_{2i}\}$, where $x'_1<x'_2<\cdots<x'_{2i}$. For each $x\in\mathbb{N}$, we set $\sigma_i(x)=1$ if $x=x'_1$ for some such soliton, and $\sigma_i(x)=0$ otherwise.

\begin{figure}
\scriptsize
\setlength{\tabcolsep}{1.0pt}
\begin{center}
\begin{tabular}{r|C{14pt}C{14pt}C{14pt}C{14pt}C{14pt}C{14pt}C{14pt}C{14pt}C{14pt}C{14pt}C{14pt}C{14pt}C{14pt}C{14pt}C{14pt}C{14pt}C{14pt}C{14pt}C{14pt}C{14pt}}
        $x$ &   0 & 1 &   2 &   3 &   4 &   5 &   6 &  7  &   8 &  9  & 10  & 11  & 12  & 13  & 14  & 15  & 16  & 17  & 18  & 19  \\
 \hline
 \hline  $\eta(x)$ & - & 1 & 1& 0 & 1 & 1 & 0 & 0 & 1 & 1 & 0 & 0 & 0 & 1 & 1 & 1 & 0 & 0 & 0 & 0 \\
 \hline
 \hline record  & $\left\{\text{-}\right\}$  &   &  &  &   &&&&&&&&&&&&&& &$\left\{0\right\}$\\
 \hline   &  &\color{gray}{1} & \color{gray}{1}& \color{gray}{0} & \color{gray}{1} &\color{gray}{1} & \color{gray}{0} & \color{gray}{0} & \color{gray}{1} & \color{gray}{1} & \color{gray}{0} & \color{gray}{0}& \color{gray}{0} & \color{gray}{1} & \color{gray}{1} & \color{gray}{1} & \color{gray}{0} & \color{gray}{0} & \color{gray}{0} & \\
 \hline size 1     &   &   &  &  $\left\{0\right.\hspace{3pt}$  & $\hspace{3pt}\left.1\right\}$  &&&&&&&&&&&&&&&\\
  \hline   &  &\color{gray}{1} & \color{gray}{1}& & &\color{gray}{1} & \color{gray}{0} & \color{gray}{0} & \color{gray}{1} & \color{gray}{1} & \color{gray}{0} & \color{gray}{0}& \color{gray}{0} & \color{gray}{1} & \color{gray}{1} & \color{gray}{1} & \color{gray}{0} & \color{gray}{0} & \color{gray}{0} & \\
 \hline size 2 &  &  & &  &  &  & $\left\{0\right.\hspace{3pt}$ & 0 & 1 & $\hspace{3pt}\left.1\right\}$ &  &  &  &  &  &  &  &  &  &  \\
   \hline   &  &\color{gray}{1} & \color{gray}{1}& & &\color{gray}{1} && &  &  & \color{gray}{0} & \color{gray}{0}& \color{gray}{0} & \color{gray}{1} & \color{gray}{1} & \color{gray}{1} & \color{gray}{0} & \color{gray}{0} & \color{gray}{0} & \\
    \hline size 3 &  &$\left\{ 1\right.\hspace{3pt}$ & 1&  &  & 1 & &  &  &  &0  &0  &$\hspace{3pt}\left.0\right\}$& &  &  &  &  &  &  \\
    \hline   &  & & & & & && &  &  &  & &  & \color{gray}{1} & \color{gray}{1} & \color{gray}{1} & \color{gray}{0} & \color{gray}{0} & \color{gray}{0} & \\
    \hline size 3 &  & &&  &  &  & &  &  &  &  &  && $\left\{1\right.\hspace{3pt}$ &1  &1  &0  &0  & $\hspace{3pt}\left.0\right\}$ &  \\
       \hline
     \hline  $\sigma_1(x)$ & - & 0 & 0& 1 & 0 & 0 & 0 & 0 & 0 & 0 & 0 & 0 & 0 & 0 & 0 & 0 & 0 & 0 & 0 & 0 \\
      \hline  $\sigma_2(x)$ & - & 0 & 0& 0 & 0 & 0 & 1 & 0 & 0 & 0 & 0 & 0 & 0 & 0 & 0 & 0 & 0 & 0 & 0 & 0 \\
       \hline  $\sigma_3(x)$ & - & 1 & 0& 0 & 0 & 0 & 0 & 0 & 0 & 0 & 0 & 0 & 0 & 1 & 0 & 0 & 0 & 0 & 0 & 0 \\
\end{tabular}
\end{center}
\caption{Identification of solitons between a pair of records. The elements shown in grey represent the parts of the configuration still being considered by the algorithm at the relevant stage. (In this example, the sizes of the identified solitons are non-decreasing, but this will not be the case in general.)}\label{solident}
\end{figure}

The next step is to place the solitons onto their effective scale where the dynamics are linear. To do this, we will appeal to the notion of a soliton `slot', which is a notion that was originally formulated in \cite{Ferrari}. Roughly speaking, a $j$-slot is a point in the configuration in which it is possible to insert a size $j$ soliton without disrupting the organisation of solitons of size $i$ for $i>j$. We will add details regarding soliton insertion in the next subsection, where we set out the spatial reconstruction of the configuration $\eta$ from the `slot decomposition'. For now, though, it is enough to note that a soliton of size $i$ admits $j$-slots as per the following diagram. (Observe that the preceding algorithm yields two possible configurations of a size $i$ soliton, depending on whether the first run consists of $1$s or $0$s.)
\medskip

\scriptsize
\setlength{\tabcolsep}{1.0pt}
\begin{center}
\begin{tabular}{r|C{24pt}C{24pt}C{24pt}C{24pt}C{24pt}C{24pt}C{24pt}C{24pt}}
        position within soliton &   1 & 2 &  \dots&$i$    &  $i+1$  & $i+2$   &\dots   & $2i$ \\
        \hline
        \hline
        soliton configuration 1 &  $\left\{ 1\right.\hspace{3pt}$ & 1&  \dots&  1& 0 & 0& \dots &$\hspace{3pt}\left.0\right\}$\\
                \hline
        soliton configuration 2 &  $\left\{ 0\right.\hspace{3pt}$ & 0&  \dots&  0& 1 & 1& \dots &$\hspace{3pt}\left.1\right\}$\\
        \hline\hline
        maximal $j$ for which position is a $j$-slot & 0&1&\dots &$i-1$& 0&1&\dots &$i-1$
                 \end{tabular}
\end{center}
\normalsize
\medskip

\noindent
We note that a $j$-slot is also a $j'$-slot for any $j'<j$. So, for each $j<i$, the number of $j$-slots contained within a size $i$ soliton is given by $2(i-j)$. In addition to the $j$-slots contained within solitons, it is also the case that each record is a $j$-slot for any $j$. Hence, since each spatial location $x\in\mathbb{Z}_+$ is either a record or at particular position within a finite soliton, by indexing sites as per the above diagram, the soliton decomposition enables us to define:
\begin{equation}\label{nudef}
\nu(x):=\sup\left\{j:\:\mbox{$x$ is a $j$-slot}\right\}.
\end{equation}
We also introduce $(S_i(x))_{i\in\mathbb{N},x\in\mathbb{Z}}$ by setting
\[S_i(x):=\sum_{x'=0}^x\mathbf{1}_{\{\nu(x')\geq i\}},\]
which gives the number of $i$-slots up to spatial location $x$. See Figure \ref{slotindex} for a continuation of our worked example, whereby we compute these quantities for the given part of the configuration.

\begin{figure}
\scriptsize
\setlength{\tabcolsep}{1.0pt}
\begin{center}
\begin{tabular}{r|C{14pt}C{14pt}C{14pt}C{14pt}C{14pt}C{14pt}C{14pt}C{14pt}C{14pt}C{14pt}C{14pt}C{14pt}C{14pt}C{14pt}C{14pt}C{14pt}C{14pt}C{14pt}C{14pt}C{14pt}}
        $x$ &   0 & 1 &   2 &   3 &   4 &   5 &   6 &  7  &   8 &  9  & 10  & 11  & 12  & 13  & 14  & 15  & 16  & 17  & 18  & 19  \\
 \hline\hline  $\eta(x)$ & - & 1 & 1& 0 & 1 & 1 & 0 & 0 & 1 & 1 & 0 & 0 & 0 & 1 & 1 & 1 & 0 & 0 & 0 & 0 \\
 \hline
 \hline record  & $\left\{\text{-}\right\}$  &   &  &  &   &&&&&&&&&&&&&& &$\left\{0\right\}$\\
 \hline size 1     &   &   &  &  $\left\{0\right.\hspace{3pt}$  & $\hspace{3pt}\left.1\right\}$  &&&&&&&&&&&&&&&\\
  \hline size 2 &  &  & &  &  &  & $\left\{0\right.\hspace{3pt}$ & 0 & 1 & $\hspace{3pt}\left.1\right\}$ &  &  &  &  &  &  &  &  &  &  \\
    \hline size 3 &  &$\left\{ 1\right.\hspace{3pt}$ & 1&  &  & 1 & &  &  &  &0  &0  &$\hspace{3pt}\left.0\right\}$&  $\left\{1\right.\hspace{3pt}$ &1  &1  &0  &0  & $\hspace{3pt}\left.0\right\}$ &  \\
     \hline \hline
      $\nu(x)$ &$ \infty$ & 0 & 1& 0 & 0 & 2 & 0 & 1 & 0 & 1 & 0 & 1 & 2 & 0 & 1 & 2& 0 & 1 & 2 & $\infty$ \\
      \hline
      \hline
       $S_1(x)$  & 1  & 1  & 2 & 2 & 2  &3&3&4&4&5&5&6&7&7&8&9&9&10&11 &12\\
       \hline
       $S_2(x)$  & 1  & 1  & 1 & 1 & 1  &2&2&2&2&2&2&2&3&3&3&4&4&4&5 &6\\
       \hline
       $S_3(x)$  & 1  & 1  & 1 & 1 & 1  &1&1&1&1&1&1&1&1&1&1&1&1&1&1 &2\\
      \end{tabular}
\end{center}
\caption{Identification of slots of given sizes.}\label{slotindex}
\end{figure}

With the above preparations in place, we are finally in a position to define the slot decomposition of a configuration $\eta\in\Omega$. In particular, this is the collection $(\zeta_i(m))_{i,m\in\mathbb{N}}$ of $\mathbb{Z}_+$-valued variables given by
\[\zeta_i(m):=\sum_{x\in\mathbb{N}}\sigma_i(x)\mathbf{1}_{\{S_i(x)=m\}}.\]
Or, in words, $\zeta_i(m)$ is the number of size $i$ solitons falling into the $m$th $i$-slot. Figure \ref{slotdecomp} shows the part of the slot decomposition that is determined by the excursion considered in our worked example (note in particular that $\zeta_i(1)=0$ for all $i$ strictly greater than the size of the largest soliton observed); the remaining entries would be obtained by concatenating row-wise the corresponding table for subsequent excursions. We will write
\begin{eqnarray}
\tilde{\Upsilon}:\Omega&\rightarrow&\bar{\Omega}\nonumber\\
\eta&\mapsto&(\zeta_i(m))_{i,m\in\mathbb{N}},\label{upsdef}
\end{eqnarray}
where $\bar{\Omega}$ is defined to be the set
\[\left\{(\zeta_i(m))_{i,m\in\mathbb{N}}\in\mathbb{Z}_+^{\mathbb{N}^2}:\:\forall m\in\mathbb{N},\:\exists i_0\mbox{ such that }\zeta_i(m)=0\mbox{ for all }i\geq i_0\right\}.\]
Furthermore, we will call $(\zeta_i(m))_{i,m\in\mathbb{N}}$ the slot decomposition of $\eta$.

\begin{figure}
\scriptsize
\begin{center}
\begin{tabular}{r|C{14pt}C{14pt}C{14pt}C{14pt}C{14pt}C{14pt}C{14pt}C{14pt}C{14pt}C{14pt}C{14pt}C{14pt}}
        $m$ &    1 &   2 &   3 &   4 &   5 &   6 &  7  &   8 &  9  & 10  & 11   \\
       \hline
       \hline
       $\zeta_1(m)$  & 0  & 1 &0 &0&0&0&0&0&0&0&0\\
       \hline
       $\zeta_2(m)$  & 0  & 1 & 0&0&0&&&&&&\\
       \hline
       $\zeta_3(m)$  & 2  &  & &&&&&&&&\\
       \hline
       $\zeta_4(m)$  & 0  &  & &&&&&&&&\\
       \hline
       $\zeta_5(m)$  & 0  &  & &&&&&&&&\\
       \hline \vdots\hspace{12pt} &\hspace{0pt}\vdots&  & &&&&&&&&
      \end{tabular}
\end{center}
\caption{The slot decomposition.}\label{slotdecomp}
\end{figure}

To complete the section, we introduce an analogue to the rescaled integrated densities of solitons $(\psi_i^N)_{i=1}^I$, as defined at \eqref{psiindef}, in terms of the positions of solitons in the slot decomposition. In particular, for $i,N\in\mathbb{N}$ and $z,t\in\mathbb{R}_+$, we define
\begin{equation}\label{barpsiindef}
\bar{\psi}_i^N(z,t):=\frac{1}{N}\sum_{m=1}^{\lfloor Nz \rfloor}\zeta_i\left(m,\lfloor Nt\rfloor\right),
\end{equation}
where for a given $\eta\in\Omega$ and $k\in\mathbb{Z}_+$, $(\zeta_i(m,k))_{i,m\in\mathbb{N}}$ is the slot decomposition of $T^k\eta$. (NB.\ We note as part of Proposition \ref{pfres} below that, for any $k\in\mathbb{Z}_+$, $T^k(\Omega)\subseteq \Omega$, and so $(\zeta_i(m,k))_{i,m\in\mathbb{N}}$ is indeed well-defined.) As will motivate the definition of our continuous state-space analogue of the slot decomposition, for each fixed $t$, we have the following connection between $\psi_i^N$ and $\bar{\psi}_i^N$ (cf.\ \eqref{barpsidef}):
\begin{equation}\label{psiin}
\bar{\psi}_i^N\approx {\psi}_i^N\circ \left(\phi_i^N\right)^{-1},
\end{equation}
where (cf.\ \eqref{phiidef})
\begin{equation}\label{phiindef}
{\phi}_i^N(u,t):=u-\sum_{j\in\mathbb{N}} 2\left(i\wedge j\right) {\psi}_j^N(u,t),
\end{equation}
and the inverse in \eqref{psiin} is the right-continuous inverse of $\phi_i^N$ taken with respect to the $u$ variable (see \eqref{rcinverse} for details). This relation will be made precise in Section \ref{hdlsec} (see Proposition \ref{gobet} and its proof in particular), but as an aid to intuition, let us briefly sketch the connection in the case $N=1$ and $t=0$. Note that in the following discussion, we drop the $t$ variable from the notation. Firstly, observe that since every spatial location is a record or a member of a soliton of a given size, we have that
\begin{equation}\label{rapprox}
x\approx R(x)+ \sum_{j\in\mathbb{N}}2j\psi_j^1(x),
\end{equation}
where $R(x)=M(x)+1$ is the number of records up to spatial location $x$. Thus, for each $i\in\mathbb{N}$, since the number of $i$-slots within each size $j$ soliton with $j>i$ is $2(j-i)$, it follows that the number of $i$-slots up to spatial location $x$, that is $S_i(x)$, satisfies
\begin{equation}\label{noslots}
S_i(x)\approx R(x)+ \sum_{j>i}2(j-i)\psi_j^1(x)\approx x- \sum_{j\in\mathbb{N}} 2\left(i\wedge j\right)\psi_j^1(x)=\phi_i^1(x).
\end{equation}
By counting the solitons according to their spatial location, or position in the slot decomposition, thus gives $\psi_i^1(x)\approx \bar{\psi}^1_i\circ\phi_i^1(x)$, which in turn yields the relation at \eqref{psiin} in this case.

\subsection{Spatial reconstruction}\label{srsubsec} Importantly, the map $\tilde{\Upsilon}$ defined at \eqref{upsdef} is a bijection, meaning that it is possible to reconstruct the spatial information about a configuration $\eta\in\Omega$ from its slot decomposition $(\zeta_i(m))_{i,m\in\mathbb{N}}$. In this subsection, we describe the procedure for doing so.

Given $(\zeta_i(m))_{i,m\in\mathbb{N}}\in\bar{\Omega}$, we start by considering the empty configuration $\eta_\infty=(\eta_\infty(x))_{x\in\mathbb{N}}$ given by $\eta_\infty(x)=0$ for all $x\in\mathbb{N}$, and go on to insert solitons one at a time until the entire slot configuration is exhausted. To begin with, we note that $\eta_\infty$ has records at every site, and so each spatial location is a $j$-slot for any $j\in\mathbb{N}$. In particular, defining $\nu_\infty$ from $\eta_\infty$ analogously to \eqref{nudef} gives $\nu_\infty(x)=\infty$ for all $x\in\mathbb{Z}_+$. We then construct the part of the configuration up to the first non-zero record as follows:
\begin{itemize}
  \item Let $i_1=\max\{i:\:\zeta_i(1)>0\}$. If the defining set is empty, then no solitons are inserted between 0 and the first non-zero record. That is, we will have in the final configuration $\eta(1)=0$, and we continue the algorithm from the record at $x=1$ as described below.
  \item If $i_1\in\mathbb{N}$, then we start by inserting $\zeta_{i_1}(1)$ size $i_1$ solitons of the form $(1,1,\dots,1,0,0,\dots,0)$, starting from spatial location 1 (i.e.\ the spatial location after the record at 0). Call the resulting configuration $\eta_{i_1}$, and moreover define $\nu_{i_1}$ from $\eta_{i_1}$ as at \eqref{nudef}. Figure \ref{reconfig} shows an example of this part of the algorithm with $i_1=3$ and $\zeta_{i_1}(1)=2$, as arising from the part of the slot decomposition shown in Figure \ref{slotdecomp}.
  \item For $j=i_1-1$, we now have a number (specifically $1+2\zeta_{i_1}(1)$) of $j$-slots from 0 to the first non-zero record. For the $m$th of these, we insert $\zeta_j(m)$ size $j$ solitons into the existing configuration, starting from the spatial location directly after the $m$th $j$-slot. Note that if $\eta_{i_1}(x)=1$ at the relevant $j$-slot, then we insert strings of the form $(0,0,\dots,0,1,1,\dots,1)$, and otherwise insert $(1,1,\dots,1,0,0,\dots,0)$. Denote the new configuration $\eta_j$, and the updated list of available slots by $\nu_j$. See Figure \ref{reconfig}.
  \item We repeat the previous step for each $j=i_1-2,\dots,1$, to give configuration $\eta_1$, and updated list of available slots $\nu_1$; again, see Figure \ref{reconfig}. If we have inserted solitons across the first $x_0$ spatial locations, it is possible to check that the first non-zero record of $\eta_1$ is at $x_0+1$, and we continue the algorithm from this spatial location as described below.
\end{itemize}
Once we have reached the first non-zero record, we delete all the entries of the slot configuration that we have already considered, and shift the remaining entries along the rows to fill in the gaps (see Figure \ref{slotdecomprem}).
We then set $\eta_\infty$ to be equal to $\eta_1$, and continue from the first non-zero record using the new slot decomposition, which we note is also an element of $\bar{\Omega}$. This results in a configuration $\eta\in\{0,1\}^\mathbb{N}$. In fact, since each record of $\eta$ corresponds to a record in the original empty configuration, and the parts of the configuration between these are finite, then it is the case that $\eta\in\Omega$. We define
\begin{eqnarray*}
\tilde{\Gamma}:\bar{\Omega}&\rightarrow&\Omega\\
(\zeta_i(m))_{i,m\in\mathbb{N}}&\mapsto&\eta
\end{eqnarray*}
by the above algorithm, and call $\eta$ the spatial reconstruction of $(\zeta_i(m))_{i,m\in\mathbb{N}}$. A key result of \cite[Section 2]{Ferrari} is (a two-sided version of) the following.

\begin{prop}\label{gyu} The map $\tilde{\Upsilon}$ is a bijection from $\Omega$ to $\bar{\Omega}$, with inverse given by $\tilde{\Gamma}$.
\end{prop}

\begin{figure}
\scriptsize
\setlength{\tabcolsep}{1.0pt}
\begin{center}
\begin{tabular}{r|C{14pt}C{14pt}C{14pt}C{14pt}C{14pt}C{14pt}C{14pt}C{14pt}C{14pt}C{14pt}C{14pt}C{14pt}C{14pt}C{14pt}C{14pt}C{14pt}C{14pt}C{14pt}C{14pt}C{14pt}}
        $x$ &   0 & 1 &   2 &   3 &   4 &   5 &   6 &  7  &   8 &  9  & 10  & 11  & 12  & 13  & 14  & 15  & 16  & 17  & 18  & 19  \\
 \hline\hline  $\eta_{\infty}(x)$ & - & 0 & 0& 0 & 0 & 0 & 0 & 0 & 0 & 0 & 0 & 0 & 0 & 0 & 0 & 0 & 0 & 0 & 0 & 0 \\
  \hline $\nu_{\infty}(x)$ & $\infty$ & $\infty$ & $\infty$&$\infty$ &$\infty$ & $\infty$ & $\infty$ &$\infty$ & $\infty$ & $\infty$ &$\infty$ & $\infty$ & $\infty$ &$\infty$ & $\infty$ & $\infty$ & $\infty$ &$\infty$ & $\infty$ & $\infty$ \\
 \hline
 \hline
 $\eta_3(x)$ & -& $\left\{1\right.\hspace{3pt}$ &1  &1  &0  &0  & $\hspace{3pt}\left.0\right\}$ &
 $\left\{1\right.\hspace{3pt}$ &1  &1  &0  &0  & $\hspace{3pt}\left.0\right\}$ &0&0&0&0&0&0&0\\
  \hline$\nu_3(x)$ &$\infty$& 0&1&2&0&1&2&0&1&2&0&1&2&$\infty$&$\infty$&$\infty$&$\infty$&$\infty$&$\infty$&$\infty$\\
      \hline   $2$-slot $\#$  & 1 &    &   &   2 &   &  &  3  &   &   & 4  &   &   & 5  &6  & 7  & 8  & 9  & 10  & 11  &12\\
      \hline
 \hline
 $\eta_2(x)$ & -& 1&1&1& $\left\{0\right.\hspace{3pt}$ &0    &1   & $\hspace{3pt}\left.1\right\}$ &0&0&0
& 1 &1  &1  &0  &0  &0&0&0&0\\
  \hline$\nu_2(x)$ &$\infty$& 0&1&2&0&1&0&1&0&1&2&0&1&2&0&1&2&$\infty$&$\infty$&$\infty$\\
   \hline   $1$-slot $\#$  & 1 &    & 2 &  3 &   &4  &  & 5  &   & 6  & 7 &   & 8  &9  &   & 10  & 11  & 12  & 13  &14\\
      \hline
 \hline
 $\eta_1(x)$ & -& 1&1&$\left\{0\right.\hspace{3pt}$ &    $\hspace{3pt}\left.1\right\}$& 1 & 0 & 0 & 1 & 1 & 0 & 0 & 0 & 1 & 1 & 1 & 0 & 0 & 0 & 0\\
 \hline$\nu_1(x)$ &$\infty$& 0&1&0&0&2&0&1&0&1&0&1&2&0&1&2&0&1&2&$\infty$\\
\end{tabular}
\end{center}
\caption{Reconstruction of the configuration up to the first non-zero record.}\label{reconfig}
\end{figure}

\begin{figure}
\scriptsize
\begin{tabular}{ccc}
\begin{tabular}{r|C{14pt}C{14pt}C{14pt}C{14pt}C{14pt}C{14pt}C{14pt}C{14pt}C{14pt}C{14pt}C{14pt}C{14pt}}
        $m$ &    1 &   2 &   3 &   4 &   5 &   6 &  7  &   8 &  9  & 10  & 11   \\
       \hline
       \hline
       $\zeta_1(m)$  & 0  & 1 &0 &0&0&0&0&0&0&0&0\\
       \hline
       $\zeta_2(m)$  & 0  & 1 & 0&0&0&&&&&&\\
       \hline
       $\zeta_3(m)$  & 2  &  & &&&&&&&&\\
       \hline
       $\zeta_4(m)$  & 0  &  & &&&&&&&&\\
       \hline
       $\zeta_5(m)$  & 0  &  & &&&&&&&&\\
       \hline \vdots\hspace{12pt} &\hspace{0pt}\vdots&  & &&&&&&&&
      \end{tabular}&\hspace{20pt}$\Rightarrow$&\hspace{20pt}\begin{tabular}{r|ccc}
        $m$ &    1 &   2 &  \dots  \\
       \hline
       \hline
       $\zeta_1(m)$  & $\zeta_1(12)$  & $\zeta_1(13)$ &\dots\\
       \hline
       $\zeta_2(m)$  &  $\zeta_2(6)$  & $\zeta_1(7)$ & \dots\\
       \hline
       $\zeta_3(m)$  & $\zeta_3(2)$  & $\zeta_3(3)$ &\dots\\
       \hline
       $\zeta_4(m)$  & $\zeta_4(2)$  & $\zeta_4(3)$ &\dots\\
       \hline
       $\zeta_5(m)$  & $\zeta_5(2)$  & $\zeta_5(3)$ &\dots\\
              \hline \vdots\hspace{12pt} &\hspace{0pt}\vdots&  \hspace{0pt}\vdots&
      \end{tabular}
      \end{tabular}
\caption{Removing slots considered up to the first non-zero record.}\label{slotdecomprem}
\end{figure}

As a final comment on the spatial reconstruction algorithm, we describe the procedure that in the scaling limit will allow us to reverse the map $\Upsilon:(\psi_i)_{i=1}^I\mapsto(\bar{\psi_i})_{i=1}^I$ introduced at \eqref{barpsidef}. As with \eqref{psiin}, this is approximate in the discrete case, but becomes rigourous asymptotically. A detailed statement again appears below (see Proposition \ref{gobet}), but we sketch the basic idea here so that the continuous state-space construction is less mysterious. For the remainder of the subsection, we fix $I\in\mathbb{N}$, and restrict our attention to configurations in $\Omega_I$ (recall this set from \eqref{omegaidef}). Ignoring the entries in the slot decomposition for slots of size $>I$ (which will all be empty), the corresponding set of slot decompositions is given by
\[\bar{\Omega}_I:=\left\{(\zeta_i(m))_{i\in\{1,\dots,I\},m\in\mathbb{N}}\in\mathbb{Z}_+^{\{1,\dots,I\}\times\mathbb{N}}\right\}.\]
For each fixed $t$, it is possible to check that
\begin{equation}\label{ooo}
{\psi}_i^N\approx \bar{\psi}_i^N\circ \gamma_i^N\circ\left(\gamma_0^N\right)^{-1},
\end{equation}
where ${\psi}_i^N$ and $\bar{\psi}_i^N$ were defined at \eqref{psiindef} and \eqref{barpsiindef}, respectively, and we define
\begin{equation}\label{ooo1}
\gamma^N_I(z,t):=z,
\end{equation}
and iteratively,
\begin{equation}\label{ooo2}
\gamma_i^N(z,t):=z+\sum_{j>i}2(j-i)\bar{\psi}_j^N\left(\gamma_j^N(z,t),t\right),\qquad i=I-1,I-2,\dots,0,
\end{equation}
with the inverse at \eqref{ooo} being the right-continuous inverse of $\gamma_0^N$ taken with respect to the $z$ variable (see \eqref{rcinverse} for details). Again, for simplicity of discussion, let us restrict to the case $N=1$ and $t=0$, and drop $t$ from the notation. Up to the $m$th non-zero record, we have $m=\gamma_I^1(m)$ $I$-slots, into which we insert $\bar{\psi}_I^1\circ\gamma_I^1(m)$ size $I$ solitons. Proceeding inductively, for $i<I$, since each soliton of size $j$ with $j>i$ admits $2(j-i)$ $i$-slots, the number of $i$-slots corresponding to the part of the configuration up to the $m$th non-zero record is given by $\gamma^1_i(m)$, and the spatial reconstruction algorithm gives that we insert into these $\bar{\psi}_i^1\circ\gamma^1_i(m)$ size $i$ solitons. Once we have inserted solitons of all sizes down to size 1, the total length of the configuration up to the $m$th non-zero record is given by $\gamma_0^1(m)$. Thus, counting size $i$ solitons according to their spatial location, or position in the slot decomposition, gives that
\begin{equation}\label{ooo3}
\psi_i^1\circ\gamma_0^1(m)=\bar{\psi}_i^1\circ\gamma_i^1(m),
\end{equation}
and applying the inverse of $\gamma_0^1$ to both sides yields \eqref{ooo}.

\subsection{Box-ball system dynamics} The key to the results of this paper is that not only can we identify solitons in the configuration, but that the dynamics on the slot decomposition corresponding to the evolution of the BBS are extremely simple. In particular, for $k\in\mathbb{Z}_+$, we introduce a shift map $\tilde{\theta}_k$ on $\bar{\Omega}$ analogous to $\theta_t$, as defined at \eqref{thetatdef}, by setting
\[\tilde{\theta}_k(\zeta_i(m))_{i,m\in\mathbb{N}}=(\zeta_i(m-ik))_{i,m\in\mathbb{N}},\]
where by convention we suppose that $\zeta_i(m-ik):=0$ for $m-ik\leq 0$. It is easy to see that $\tilde{\theta}_k(\bar{\Omega})\subseteq \bar{\Omega}$ for any $k\in\mathbb{Z}_+$, and so the space of slot decompositions is preserved by these shift mappings. Moreover, we have the following consequence of \cite[Theorem 2.1]{Ferrari}, which will ultimately lead to the flow equation for the limiting model, as given at \eqref{BBSflow}. NB.\ The result of \cite{Ferrari} is more complicated than we state here due to the more general bi-infinite configurations considered in the latter article, for which it is necessary to be careful about the solitons crossing the origin.

\begin{prop}[{cf.\ \cite[Theorem 2.1]{Ferrari}}] \label{pfres}
For $\eta\in\Omega$ and $k\in\mathbb{Z}_+$,
\[T^k\eta=\tilde{\Upsilon}^{-1}\circ\tilde{\theta}_k\circ\tilde{\Upsilon}( \eta).\]
In particular, $T^k(\Omega)\subseteq\Omega$ for any $k\in\mathbb{Z}_+$.
\end{prop}

\section{Continuous soliton decomposition and dynamics}\label{contsec}

Using the slot decomposition and spatial reconstruction operators of the previous section as our inspiration, in this section we introduce our continuous state-space model. The definitions of the continuous soliton decomposition and spatial reconstruction appear in Subsections \ref{31} and \ref{32}, respectively, and we check that they are inverses of each other on an appropriate domain in Subsection \ref{33}. Moreover, in Subsection \ref{34}, we check that the domain $\mathcal{D}$ introduced at \eqref{ddef} is preserved by the flow at \eqref{BBSflow}. Throughout we fix $I\in\mathbb{N}$.

\subsection{Soliton decomposition operator}\label{31}

Our continuous soliton decomposition will be defined on the set
\begin{equation}\label{fdef}
\mathcal{F}:=\left\{(\psi_i)_{i=1}^I\in\mathcal{C}^I:\phi_I\in \mathcal{C}^{\uparrow}\right\},
\end{equation}
where we recall the definitions of $\mathcal{C}$, $\mathcal{C}^\uparrow$ and $\phi_i$ from \eqref{cdef}, \eqref{cupdef} and \eqref{phiidef}, respectively. As commented in the introduction, we will think of $\psi_i=(\psi_i(u))_{u\in\mathbb{R}_+}$ as representing the integrated density of size $i$ solitons in space, as indexed by $u\in\mathbb{R}_+$, cf.\ the discrete definition of $\psi_i^N$ at \eqref{psiindef}. Although we only assume that $\phi_I\in\mathcal{C}^\uparrow$, the following lemma shows that this is equivalent to assuming that $\phi_i\in\mathcal{C}^\uparrow$ for every $i\in\{1,2,\dots,I\}$.

\begin{lem}\label{philem}
Let $\psi\in \mathcal{F}$. For $i=1,2,\dots,I$, it holds that $\phi_i\in\mathcal{C}^\uparrow$.
\end{lem}
\begin{proof} By definition, $\phi_i$ is clearly continuous and satisfies $\phi_i(0)=0$. For any $u_1\geq u_2$, it also holds that $\phi_i(u_1)-\phi_i(u_2)\geq \phi_I(u_1)-\phi_I(u_2)$. Since $\phi_I\in\mathcal{C}^\uparrow$ by assumption, the result follows.
\end{proof}

By analogy with the discussion at the end of Subsection \ref{soldecompsec}, and \eqref{noslots} in particular, we think of $\phi_i(u)$ as the integrated $i$-slot density up to spatial location $u$; that is, the effective distance from 0 to $u$ in terms of $i$-slots. The previous lemma shows that this change of scale from space to slot location is smooth in the sense that, for each $i$, $\phi_i:\mathbb{R}_+\rightarrow\mathbb{R}_+$ is a homeomorphism, with inverse satisfying $\phi_i^{-1}\in\mathcal{C}^\uparrow$. This basic property allows us to define the integrated density of solitons on their effective scale, as per the following definition, cf.\ \eqref{psiin}. The map $\Upsilon$ will play the role that $\tilde{\Upsilon}$, as defined at \eqref{upsdef}, did in the discrete case. In the language of integrable systems, we might think of $\Upsilon$ as the `scattering map' in the present context.

\begin{df} For $\psi\in \mathcal{F}$, define $\Upsilon(\psi):=\bar{\psi}$, where  $\bar{\psi}=(\bar{\psi}_i)_{i=1}^I$ is given by
\[\bar{\psi}_i:=\psi_i\circ\phi_i^{-1}.\]
We say that $\Upsilon(\psi)$ is the \emph{effective scaling} of $\psi$.
\end{df}

To complete the subsection, we make the following observation about the image of the map $\Upsilon$.

\begin{prop}\label{into} It holds that $\Upsilon(\mathcal{F})\subseteq \mathcal{C}^I$.
\end{prop}
\begin{proof} Let $\psi\in\mathcal{F}$, and set $\bar{\psi}:=\Upsilon(\psi)$. As an easy consequence of Lemma \ref{philem}, we have that $\phi_i^{-1}\in\mathcal{C}^\uparrow$. Hence, since $\psi_i\in\mathcal{C}$, we obtain that $\bar{\psi}_i=\psi_i\circ\phi_i^{-1}\in\mathcal{C}$.
\end{proof}

\subsection{Spatial reconstruction operator}\label{32}

Towards inverting the effective scaling map $\Upsilon$, we introduce maps $(\gamma_i)_{i=0}^I$ that will correspond to the discrete versions of \eqref{ooo1} and \eqref{ooo2}. In particular, for $\bar{\psi}\in \mathcal{C}^I$, we define $\gamma_I(z):=z$, and iteratively,
\[\gamma_i(z):=z+\sum_{j>i}2\left(j-i\right)\bar{\psi}_j\circ \gamma_j(z),\qquad i=I-1,I-2,\dots,0.\]
We have the following basic property of these maps.

\begin{lem}\label{ttt1}
Let $\bar{\psi}\in \mathcal{C}^I$. For $i=0,1,\dots,I$, it holds that $\gamma_i\in\mathcal{C}^\uparrow$.
\end{lem}
\begin{proof} The result is a straightforward consequence of the definition.
\end{proof}

As a result of this lemma, we have that $\gamma_0^{-1}\in \mathcal{C}^\uparrow$ for any $\bar{\psi}\in \mathcal{C}^I$. This allows us to formulate the following definition, which is based on the corresponding connection between the slot decomposition and its spatial reconstruction for the discrete system, see \eqref{ooo}. As such, to again apply a term from integrable systems, we might think of the map that is introduced as the `inverse scattering map'. That it is indeed the inverse of the scattering map $\Upsilon$ is checked below in Proposition \ref{upsinv}.

\begin{df}
For $\bar{\psi}\in \mathcal{C}^I$, define  $\Gamma(\bar{\psi}):=\tilde{\psi}$, where  $\tilde{\psi}=(\tilde{\psi}_i)_{i=1}^I$ is given by
\[\tilde{\psi}_i:=\bar{\psi}_i\circ\gamma_i\circ\gamma_0^{-1}.\]
We say that $\Gamma(\bar{\psi})$ is the \emph{spatial reconstruction} of $\bar{\psi}$.
\end{df}

Concerning the image of $\Gamma$, we have the following result, for which we recall the definition of $\mathcal{F}$ from \eqref{fdef}.

\begin{prop}\label{into2} It holds that $\Gamma(\mathcal{C}^I)\subseteq\mathcal{F}$.
\end{prop}
\begin{proof} Let $\bar{\psi}\in\mathcal{C}^I$, and set $\tilde{\psi}:=\Gamma(\bar{\psi})$. From Lemma \ref{ttt1}, it is clear that $\tilde{\psi}_i\in\mathcal{C}$ for each $i$. It remains to check that $\tilde{\phi}_I\in\mathcal{C}^{\uparrow}$, where
\[\tilde{\phi}_I(u):=u-\sum_{j=1}^I2j\tilde{\psi}_j(u).\]
We first observe that
\[\tilde{\phi}_I\circ\gamma_0(z)=z+\sum_{j=1}^I2 j\bar{\psi}_j\circ \gamma_j(z)-\sum_{j=1}^I2j\bar{\psi}_j\circ\gamma_j(z)=z.\]
Together with Lemma \ref{ttt1}, this implies that $\tilde{\phi}_I=\gamma_0^{-1}\in\mathcal{C}^\uparrow$, as desired.
\end{proof}

\subsection{Invertibility of soliton decomposition operator}\label{33}

Connecting the constructions of the previous two subsections, we have the following continuous analogue of Proposition \ref{gyu}.

\begin{prop}\label{upsinv} The map $\Upsilon$ is a bijection from $\mathcal{F}$ to $\mathcal{C}^I$, with inverse given by $\Gamma$.
\end{prop}
\begin{proof} We start by checking that if $\psi\in\mathcal{F}$, then $\Gamma\circ\Upsilon(\psi)=\psi$. Let $\psi\in\mathcal{F}$, and set $\bar{\psi}:=\Upsilon(\psi)$. By Proposition \ref{into}, it holds that $\bar{\psi}\in\mathcal{C}^I$, and so $\tilde{\psi}:=\Gamma(\bar{\psi})$ is well-defined. We need to show that $\tilde{\psi}=\psi$. To this end, we first establish that
\begin{equation}\label{ttt}
\gamma_i\circ\phi_I=\phi_i,\qquad i=0,1,\dots,I,
\end{equation}
where $\phi_0$ is defined by setting $\phi_0(u):=u$. We will prove this by induction. Since $\gamma_I(u)=u$, the case $i=I$ is obvious. Suppose we have proven the equality for $I,I-1,\dots,i+1$. We then have that
\begin{align*}
\gamma_i\circ\phi_I(u)&=\phi_I(u)+\sum_{j>i}2(j-i)\bar{\psi}_j\circ\gamma_j\circ\phi_I(u)\\
&=u-\sum_{j}2j{\psi}_j(u)+\sum_{j>i}2(j-i){\psi}_j\circ\phi_j^{-1}\circ\gamma_j\circ\phi_I(u)\\
&=u-\sum_{j}2(j\wedge i){\psi}_j(u)\\
&=\phi_i(u),
\end{align*}
where to deduce the third equality we apply the inductive hypothesis. Thus we have established \eqref{ttt}, and applying this twice (once for $i$, and once for $0$), it follows that
\[\tilde{\psi}_i=\bar{\psi}_i\circ\gamma_i\circ\gamma_0^{-1}={\psi}_i\circ \phi_i^{-1}\circ\gamma_i\circ\gamma_0^{-1}=\psi_i\circ\phi_I^{-1}\circ\gamma_0^{-1}=\psi_i,\]
which gives the injectivity of $\Upsilon$.

To complete the proof, we will show that if $\bar{\psi}\in\mathcal{C}^I$, then $\Upsilon\circ\Gamma(\bar{\psi})=\bar{\psi}$. Let $\bar{\psi}\in\mathcal{C}^I$, and set $\tilde{\psi}:=\Gamma(\bar{\psi})$. By Proposition \ref{into2}, we know that $\tilde{\psi}\in\mathcal{F}$, and so $\hat{\psi}:=\Upsilon(\tilde{\psi})$ is well-defined. Toward showing that $\hat{\psi}=\bar{\psi}$, we start by checking that
\[\tilde{\phi}_i\circ\gamma_0=\gamma_i,\qquad i=0,1,\dots,I,\]
where $\tilde{\phi}_i$ is defined by setting
\[\tilde{\phi}_i(u):=u-\sum_{j=1}^I2(j\wedge i)\tilde{\psi}_j(u).\]
Indeed, as in the proof of Proposition \ref{into2}, we deduce that
\begin{align*}
\tilde{\phi}_i\circ\gamma_0(z)&=z+\sum_{j=1}^I2 j\bar{\psi}_j\circ \gamma_j(z)-\sum_{j=1}^I2(j\wedge i)\bar{\psi}_j\circ\gamma_j(z)\\
&=z+\sum_{j>i}2 (j-i)\bar{\psi}_j\circ \gamma_j(z)\\
&=\gamma_i(u).
\end{align*}
Hence
\[\hat{\psi}_i=\tilde{\psi}_i\circ\tilde{\phi}_i^{-1}=\bar{\psi}_i\circ\gamma_i\circ\gamma_0^{-1}\circ\tilde{\phi}_i^{-1}=\bar{\psi}_i,\]
and so we are done.
\end{proof}

\subsection{Time-preserved density condition}\label{34} Recall domain $\mathcal{D}$ from \eqref{ddef}, which we claimed represented a set of integrated densities satisfying a time-preserved density condition. In this subsection, we will show that the corresponding set of integrated densities on the effective scale is given by
\begin{equation}\label{dbardef}
\bar{\mathcal{D}}:=\left\{(\bar{\psi}_i)_{i=1}^I\in\mathcal{C}^I: \sum_{i=1}^I i \sup_{z_1,z_2}\frac{\bar{\psi}_i(z_1)-\bar{\psi}_i(z_2)}{z_1-z_2}<\frac{1}{2}\right\},
\end{equation}
and that this set is indeed preserved by the shift operator $\theta_t$ from \eqref{thetatdef}, which can be thought of as describing the `free evolution' of solitons. In particular, this implies that, for any initial condition in $\mathcal{D}$, the flow at \eqref{BBSflow} remains in $\mathcal{D}$ for all time.

\begin{prop}\label{tpprop}
(a) The map $\Upsilon$ is a bijection from $\mathcal{D}$ to $\bar{\mathcal{D}}$.\\
(b) It holds that $\theta_t(\bar{\mathcal{D}})\subseteq\bar{\mathcal{D}}$ for any $t\geq 0$.\\
(c) It holds that $\Upsilon^{-1}\circ\theta_t\circ\Upsilon(\mathcal{D})\subseteq{\mathcal{D}}$ for any $t\geq 0$.
\end{prop}
\begin{proof} Let $\psi\in\mathcal{D}$, and set $\bar{\psi}:=\Upsilon(\psi)$. From Lemma \ref{into2}, we know that $\bar{\psi}\in\mathcal{C}^{I}$. Moreover, since $\phi_i\in\mathcal{C}^\uparrow$ for each $i$, it holds that
\begin{align*}
\sum_{i=1}^I i \sup_{z_1,z_2}\frac{\bar{\psi}_i(z_1)-\bar{\psi}_i(z_2)}{z_1-z_2}&=
\sum_{i=1}^I i \sup_{u_1,u_2}\frac{\bar{\psi}_i\circ\phi_i(u_1)-\bar{\psi}_i\circ\phi_i(u_2)}{\phi_i(u_1)-\phi_i(u_2)}\\
&=\sum_{i=1}^I i \sup_{u_1,u_2}\frac{{\psi}_i(u_1)-{\psi}_i(u_2)}{\phi_i(u_1)-\phi_i(u_2)}\\
&<\frac{1}{2},
\end{align*}
and so $\bar{\psi}\in\bar{\mathcal{D}}$. The converse (i.e.\ $\bar{\psi}\in\bar{\mathcal{D}}$ implies $\Gamma(\bar{\psi})\in\mathcal{D}$) is proved in the same way. This establishes (a).

For part (b), let $\bar{\psi}\in\bar{\mathcal{D}}$. We then have that, for any $t\geq 0$,
\begin{align*}
\lefteqn{\sum_{i=1}^I i \sup_{z_1,z_2}\frac{(\theta_t\circ\bar{\psi})_i(z_1)-(\theta_t\circ\bar{\psi})_i(z_2)}{z_1-z_2}}\\
&=\sum_{i=1}^I i \sup_{z_1,z_2}\frac{\bar{\psi}_i\left((z_1-it)\vee0\right)-\bar{\psi}_i\left((z_2-it)\vee0\right)}{z_1-z_2}\\
&=\sum_{i=1}^I i \sup_{z_1,z_2}\frac{\bar{\psi}_i(z_1)-\bar{\psi}_i(z_2)}{z_1-z_2}\\
&<\frac{1}{2},
\end{align*}
which means that $\theta_t(\bar{\psi})\in\bar{\mathcal{D}}$.

Finally, for part (c), we note that, for any $t\geq 0$,
\[\Upsilon^{-1}\circ\theta_t\circ\Upsilon(\mathcal{D})=\Upsilon^{-1}\circ\theta_t(\bar{\mathcal{D}})\subseteq\Upsilon^{-1}(\bar{\mathcal{D}})=\mathcal{D},\]
and thus the proof is complete.
\end{proof}

\section{Generalized hydrodynamic limit for integrated density}\label{hdlsec}

With both the discrete and continuous soliton decompositions now introduced, we can proceed to establishing the generalized hydrodynamic limit for integrated soliton densities of Theorem \ref{mainthm}. The key to the proof of this result is the following proposition, which essentially shows that convergence under scaling of soliton densities in space is equivalent to the corresponding convergence with respect to effective distances. Note that, for any non-decreasing, right-continuous $f:\mathbb{R}_+\rightarrow\mathbb{R}_+$, we define the right-continuous inverse $f^{-1}:\mathbb{R}_+\rightarrow[0,\infty]$ by setting
\begin{equation}\label{rcinverse}
f^{-1}(z):=\inf\left\{u:\:f(u)> z\right\},
\end{equation}
with the convention that $\inf\emptyset=\infty$.

\begin{prop}\label{gobet} For each $N\in\mathbb{N}$, suppose $\eta$ is a random configuration with distribution $\mathbf{P}_N$ supported on $\Omega_I$ for some $I\in \mathbb{N}$ (independent of $N$). Set $\psi_i^N:=\psi_i^N(\cdot,0)$ and $\bar{\psi}_i^N:=\bar{\psi}_i^N(\cdot,0)$, as per the definitions at \eqref{psiindef} and \eqref{barpsiindef}. Moreover, let $\psi\in\mathcal{D}$, and define $\bar{\psi}:=\Upsilon(\psi)$. The following two conditions are equivalent.\\
(a) For every $\varepsilon>0$ and $u_0\in(0,\infty)$,
\begin{equation}\label{a}
\lim_{N\rightarrow\infty}\mathbf{P}_N\left(\sup_{i\in\{1,2,\dots,I\}}\sup_{u\leq u_0}\left|\psi_i^N(u)-\psi_i(u)\right|>\varepsilon\right)=0.
\end{equation}
(b) For every $\varepsilon>0$ and $z_0\in(0,\infty)$,
\begin{equation}\label{b}
\lim_{N\rightarrow\infty}\mathbf{P}_N\left(\sup_{i\in\{1,2,\dots,I\}}\sup_{z\leq z_0}\left|\bar{\psi}_i^N(z)-\bar{\psi}_i(z)\right|>\varepsilon\right)=0.
\end{equation}
\end{prop}
\begin{proof} We start by checking that (a) implies (b). Let $Nu$ be an integer. Since each spatial site is either a record or contained in a soliton, it must hold that
\begin{equation}\label{star}
\left|Nu-R(Nu)-N\sum_{j=1}^I2j\psi_j^N(u)\right|\leq I(2I-1)\leq 2I^2,
\end{equation}
where we recall $R(x)$ is the number of records up to spatial location $x$, cf.\ \eqref{rapprox}. In particular, the bound comes from the fact that a size $j$ soliton that has its first site in $\{1,2,\dots,Nu\}$ can have at most $2j-1$ sites outside this interval. We similarly have that the number of $i$-slots up to spatial location $Nu$, that is $S_i(Nu)$, satisfies
\[\left|S_i(Nu)-R(Nu)-N\sum_{j>i}2(j-i)\psi_j^N(u)\right|\leq 2(I-1)^2.\]
Putting these bounds together gives
\begin{equation}\label{star2}
\left|S_i(Nu)-N\phi_i^N(u)\right|\leq 4I^2,
\end{equation}
where $\phi_i^N:=\phi_i^N(\cdot,0)$, with the latter function being defined at \eqref{phiindef}, cf.\ \eqref{noslots}. Now, applying \eqref{a} in conjunction with the density condition that holds on $\mathcal{D}$ allows it to be deduced that there exists an $\varepsilon>0$ such that, for $u_0\in(0,\infty)$ and $\delta\in(0,1)$,
\[\lim_{N\rightarrow\infty}\mathbf{P}_N\left(\sup_{u\leq u_0}\sum_{j=1}^I2j\left(\psi_j^N(u+\delta)-\psi_j^N(u)\right)\geq \delta(1-\varepsilon)\right)=0,\]
where we have also applied the fact that $\phi_j(u+\delta)-\phi_j(u)\leq \delta$ for each $j$.
On the complement of the event in the above probability, for any $u\leq u_0$ and $i\in\{1,2,\dots,I\}$, we have by definition that $N \phi_i^N(u+\delta)-N \phi_i^N(u)\geq \varepsilon\delta N$. In particular, this establishes that, for large enough $N$,
\[N\phi_i^N(u+\delta)\geq N\phi_i^N(u)+\varepsilon\delta N\geq S_i(Nu),\]
where we have applied \eqref{star2} to obtain the second inequality. One similarly deduces that, for large enough $N$,
\[N\phi_i^N(u)\leq N\phi_i^N(u+\delta) -\varepsilon\delta N\leq S_i(N(u+\delta)).\]
By considering the numbers of solitons in the relevant slots, it follows that, for large enough $N$ and $u\leq u_0-\delta$,
\[\psi_i^N\left((u-\delta)\vee 0\right)\leq \bar{\psi}_i^N\circ\phi_i^N(u)\leq \bar{\psi}_i^N\left(N^{-1} S_i\left(N(u+\delta)\right)\right)\leq\psi_i^N(u+2\delta),\]
cf.\ \eqref{psiin}, where for the final inequality we apply that
\[R(N(u+2\delta))-R(N(u+\delta))\geq \lfloor N(u+2\delta)\rfloor -\lfloor N(u+\delta)\rfloor-\delta(1-\varepsilon)N-4I^2\geq 1,\]
which is a consequence of \eqref{star} and, since every record is an $i$-slot, implies that any soliton that is in the same slot as a soliton with its first site in $\{1,2,\dots,N(u+\delta)\}$ must be located spatially in $\{1,2,\dots,N(u+2\delta)\}$. Combining the above argument with \eqref{a} thus yields, for $u_0\in(0,\infty)$ and any $\varepsilon,\delta\in(0,1)$,
\begin{align*}
\lefteqn{\lim_{N\rightarrow\infty}\mathbf{P}_N\left(\psi_i((u-\delta)\vee0)-\varepsilon\leq \bar{\psi}_i^N\circ\phi_i^N(u)\leq \psi_i(u+\delta)+\varepsilon\mbox{ for all }u\leq u_0\right)}\\
&\hspace{330pt}=1.
\end{align*}
Taking $\delta$ arbitrarily small, and appealing to the continuity of $\psi_i$, this implies that, for every $\varepsilon\in(0,1)$ and $u_0\in(0,\infty)$,
\begin{equation}\label{hui}
\lim_{N\rightarrow\infty}\mathbf{P}_N\left(\sup_{u\leq u_0}\left|\bar{\psi}_i^N\circ\phi_i^N(u)-\psi_i(u)\right|>\varepsilon\right)=0.
\end{equation}
Next, by the definitions of $\phi_i^N$ and $\phi_i$, it is clear that \eqref{a} implies, for every $\varepsilon>0$ and $u_0\in(0,\infty)$,
\[\lim_{N\rightarrow\infty}\mathbf{P}_N\left(\sup_{i\in\{1,2,\dots,I\}}\sup_{u\leq u_0}\left|\phi_i^N(u)-\phi_i(u)\right|>\varepsilon\right)=0.\]
Moreover, since $\phi_i\in\mathcal{C}^\uparrow$ (by Lemma \ref{philem}), it follows that, for every $\varepsilon>0$ and $z_0\in(0,\infty)$,
\begin{equation}\label{rrr1}
\lim_{N\rightarrow\infty}\mathbf{P}_N\left(\sup_{i\in\{1,2,\dots,I\}}\sup_{z\leq z_0}\left|(\phi_i^N)^{-1}(z)-\phi_i^{-1}(z)\right|>\varepsilon\right)=0,
\end{equation}
and also
\begin{equation}\label{rrr2}
\lim_{N\rightarrow\infty}\mathbf{P}_N\left(\sup_{i\in\{1,2,\dots,I\}}\sup_{z\leq z_0}\left|\phi_i^N\circ(\phi_i^N)^{-1}(z)-z\right|>\varepsilon\right)=0.
\end{equation}
Composing $\bar{\psi}_i^N\circ\phi_i^N$ with $(\phi_i^N)^{-1}$, ${\psi}_i$ with $\phi_i^{-1}$, applying \eqref{hui} with \eqref{rrr1} yields that, for every $\varepsilon>0$ and $z_0\in(0,\infty)$,
\[\lim_{N\rightarrow\infty}\mathbf{P}_N\left(\sup_{i\in\{1,2,\dots,I\}}\sup_{z\leq z_0}\left|\bar{\psi}_i^N\circ\phi_i^N\circ(\phi_i^N)^{-1}(z)-\bar{\psi}_i(z)\right|>\varepsilon\right)=0.\]
Combining this limit with \eqref{rrr2}, the monotonicity of $\bar{\psi}_i^N$, and the continuity of $\bar{\psi}_i$ gives \eqref{b}.

It remains to check that (b) implies (a). Writing $\gamma_i^N:=\gamma_i^N(\cdot,0)$, where $\gamma_i^N(z,t)$ was defined at \eqref{ooo2}, and applying the definition of $\gamma_i$ from Subsection \ref{32}, it is clear that \eqref{b} implies, for every $\varepsilon>0$ and $z_0\in(0,\infty)$,
\begin{equation}\label{c}
\lim_{N\rightarrow\infty}\mathbf{P}_N\left(\sup_{i\in\{0,1,\dots,I\}}\sup_{z\leq z_0}\left|\gamma_i^N(z)-\gamma_i(z)\right|>\varepsilon\right)=0.
\end{equation}
Moreover, since $\gamma_i\in\mathcal{C}^\uparrow$ (by Lemma \ref{ttt1}), it follows that, for every $\varepsilon>0$ and $u_0\in(0,\infty)$,
\begin{equation}\label{d}
\lim_{N\rightarrow\infty}\mathbf{P}_N\left(\sup_{i\in\{0,1,\dots,I\}}\sup_{u\leq u_0}\left|(\gamma_i^N)^{-1}(u)-\gamma_i^{-1}(u)\right|>\varepsilon\right)=0,
\end{equation}
and also
\begin{equation}\label{d2}
\lim_{N\rightarrow\infty}\mathbf{P}_N\left(\sup_{i\in\{0,1,\dots,I\}}\sup_{u\leq u_0}\left|\gamma_i^N\circ(\gamma_i^N)^{-1}(u)-u\right|>\varepsilon\right)=0.
\end{equation}
Now, if $Nz$ is an integer, then by the spatial reconstruction algorithm of Subsection \ref{srsubsec}, we have that the first $N\gamma^N_0(z)$ spatial sites contain precisely the solitons of size $i$ in the first $N\gamma^N_i(z)$ slots of the relevant size, for each $i=1,2,\dots,I$ (cf.\ the discussion at the end of Subsection \ref{srsubsec}). Thus, for any $z$, we have that
\[\bar{\psi}_i^N\circ\gamma_i^N\left(\frac{\lfloor Nz \rfloor}{N}\right)\leq {\psi}_i^N\circ\gamma_0^N\left(z\right)\leq \bar{\psi}_i^N\circ\gamma_i^N\left(\frac{\lceil Nz \rceil}{N}\right),\]
which makes precise \eqref{ooo3}. Combining these inequalities with \eqref{b} and \eqref{c} yields that, for every $\varepsilon>0$ and $z_0\in(0,\infty)$,
\[\lim_{N\rightarrow\infty}\mathbf{P}_N\left(\sup_{i\in\{0,1,\dots,I\}}\sup_{z\leq z_0}\left|{\psi}_i^N\circ\gamma_0^N\left(z\right)-\bar{\psi}_i\circ\gamma_i(z)\right|>\varepsilon\right)=0.\]
Composing ${\psi}_i^N\circ\gamma_0^N$ with $(\gamma_0^N)^{-1}$, $\bar{\psi}_i\circ\gamma_i$ with $\gamma_0^{-1}$, applying the above limit with \eqref{d}, \eqref{d2}, the monotonicity of $\psi_i^N$, and the continuity of $\psi_i$ (similarly to the end of the first part of the proof) gives \eqref{a}. Thus the proof is complete.
\end{proof}

We finish the section by proving Theorem \ref{mainthm}.

\begin{proof}[Proof of Theorem \ref{mainthm}] By the assumption at \eqref{initconv} and Proposition \ref{gobet}, we have that, for every $\varepsilon>0$ and $z_0\in(0,\infty)$,
\[\lim_{N\rightarrow\infty}\mathbf{P}_N\left(\sup_{i\in\{1,2,\dots,I\}}\sup_{z\leq z_0}\left|\bar{\psi}_i^N(z,0)-(\Upsilon\circ\psi^0)_i(z)\right|>\varepsilon\right)=0.\]
Since in terms of the slot positions the dynamics of the discrete solitons are linear (recall Proposition \ref{pfres}), it readily follows that, for every $\varepsilon>0$, $t\in(0,\infty)$ and $z_0\in(0,\infty)$,
\begin{equation}\label{thetacon}
\lim_{N\rightarrow\infty}\mathbf{P}_N\left(\sup_{i\in\{1,2,\dots,I\}}\sup_{z\leq z_0}\left|\bar{\psi}_i^N(z,t)-(\theta_t\circ\Upsilon\circ\psi^0)_i(z)\right|>\varepsilon\right)=0.
\end{equation}
Now, from Proposition \ref{tpprop}, we know that $\theta_t\circ\Upsilon(\psi^0)\in\bar{\mathcal{D}}$. Hence, we can apply Proposition \ref{gobet} again to deduce that, for every $\varepsilon>0$, $t\in(0,\infty)$ and $u_0\in(0,\infty)$,
\[\lim_{N\rightarrow\infty}\mathbf{P}_N\left(\sup_{i\in\{1,2,\dots,I\}}\sup_{u\leq u_0}\left|{\psi}_i^N(u,t)-(\Upsilon^{-1}\circ\theta_t\circ\Upsilon\circ\psi^0)_i(u)\right|>\varepsilon\right)=0,\]
which completes the proof.
\end{proof}

\section{Effective speeds}\label{msec}

The goal of this subsection is to prove the following lemma, which ensures that, for densities that are not too large, the effective speeds of solitons are well-defined by \eqref{effv}. Since the proof is a relatively straightforward technical exercise in linear algebra, a reader interested in the essence of the generalized hydrodynamics story might wish to skip it.

\begin{lem}\label{minvert} Fix $I\in\mathbb{N}$. For $\rho \in [0,1]^I$ satisfying $\sum_{i=1}^I 2i \rho_i <1$, the matrix $M=M(\rho)$ defined at \eqref{matrix} is invertible.
\end{lem}
\begin{proof} In the proof, we will write $a_{ij}:=2(i\wedge j)\rho_j$. Moreover, since the case $I=1$ is obvious, we will assume $I\geq 2$. Now, $M$ is given by
\[\left(
  \begin{array}{ccccc}
    1-\sum_{j\neq 1}a_{1j} & a_{12}                 & a_{13} & \dots  & a_{1I} \\
                    a_{21} & 1-\sum_{j\neq 2}a_{2j} & a_{23} &  \dots & a_{2I} \\
    a_{31}                 & a_{32}&  1-\sum_{j\neq 3}a_{3j} & \ddots& \vdots \\
                      \vdots &\vdots  &   \ddots  &\ddots &a_{(I-1)I}\\
    a_{I1} & a_{I2} &  \dots&  a_{I(I-1)}& 1-\sum_{j\neq I}a_{Ij} \\
  \end{array}
\right).\]
Since each of the rows sum to one, $M$ has the same determinant as the following matrix:
\[M':=\left(
  \begin{array}{ccccc}
    1 & a_{12} & a_{13} & \dots & a_{1I} \\
    1 & 1-\sum_{j\neq 2}a_{2j} & a_{23} &  \dots & a_{2I} \\
    1 & a_{32} &  1-\sum_{j\neq 3}a_{3j} & \ddots & \vdots \\
   \vdots & \vdots &\ddots  & \ddots &a_{(I-1)I} \\
    1 & a_{I2} & \dots &  a_{I(I-1)} & 1-\sum_{j\neq I}a_{Ij} \\
  \end{array}
\right).\]
Subtracting $2j\rho_j$ times the first column from the $j$th column, $j=2,3,\dots,I$, we find that the determinant of $M$ is also the same as that of a matrix of the form:
\[M''=\left(
  \begin{array}{ccccc}
    1 & \leq0 & \leq0 & \dots & \leq0 \\
    1 & d_2 & \leq0 &  \dots & \leq0 \\
    1 & 0 &  d_3 &\ddots  & \vdots \\
   \vdots & \vdots &\ddots  & \ddots &  \leq0  \\
    1 & 0 & \dots &  0& d_I \\
  \end{array}
\right),\]
where we write $d_i:=1-\sum_ja_{ij}$, and here and throughout the remainder of the proof, denote by the symbol `$\leq 0$' a number that is non-positive, but the particular value of which is unimportant (and similarly for `$\geq0$'). Note that, by assumption, $d_i>0$ for each $i$.

Now, write $M^{(I)}$ for a matrix of the form of $M''$. We will demonstrate that $\mathrm{det}(M^{(I)})>0$. Clearly $\mathrm{det}(M^{(2)})\geq d_2>0$. Hence we will now assume $I\geq 3$, and expand the determinant along the bottom row of the matrix to deduce that
\begin{equation}\label{mdecomp}
\mathrm{det}\left(M^{(I)}\right)=d_I \mathrm{det}\left(M^{(I-1)}\right)+(-1)^{I-1}\mathrm{det}\left(N^{(I-1)}\right),
\end{equation}
where the matrix $M^{(I-1)}$ is an $(I-1)\times (I-1)$ matrix with the same form as $M^{(I)}$ (but a smaller dimension), and $N^{(I-1)}$ is a matrix of the form:
\[N^{(I-1)}=\left(
  \begin{array}{ccccc}
    \leq0 & \leq0 & \leq0 & \dots & \leq0 \\
    d_2 &  \leq0& \leq0 &  \dots & \leq0 \\
    0 &  d_3 & \leq0 &\ddots  & \vdots \\
   \vdots & \vdots &\ddots  & \ddots &  \leq0  \\
    0 &\dots &0 & d_{I-1}&\leq0  \\
  \end{array}
\right).\]
We next check that
\begin{equation}\label{nineq}
(-1)^{I-1}\mathrm{det}\left(N^{(I-1)}\right)\geq 0
\end{equation}
by induction. This is clear for $I=3$. Moreover, for $I>3$, by expanding the determinant along the bottom row and applying the inductive hypothesis, we see that
\begin{align*}
\lefteqn{(-1)^{I-1}\mathrm{det}\left(N^{(I-1)}\right)}\\
&=(-1)^{I-1}\left(\leq 0\right) \mathrm{det}\left(N^{(I-2)}\right)-(-1)^{I-1}d_{I-1}\mathrm{det}\left(\tilde{N}^{(I-2)}\right)\\
&=\left(\geq 0\right)(-1)^{I-2} \mathrm{det}\left(N^{(I-2)}\right)+d_{I-1}(-1)^{I-2}\mathrm{det}\left(\tilde{N}^{(I-2)}\right)\\
&\geq0,
\end{align*}
where we have written $N^{(I-2)}$ and $\tilde{N}^{(I-2)}$ for matrices of the same form as $N^{(I-1)}$ (but a smaller dimension). This confirms the inequality at \eqref{nineq}. Returning to \eqref{mdecomp}, we thus find that $\mathrm{det}(M^{(I)})\geq d_{I}d_{I-1}\dots d_2>0$, and so we are done.
\end{proof}

\section{Partial differential equation descriptions}\label{pdesec}

The only significant hurdle remaining is to show that, under appropriate regularity conditions, the dynamics of the continuous state-space system are described by the partial differential equations \eqref{PDErho} and \eqref{PDEpsi}. We will do this in Subsection \ref{pdesubsec}. Prior to this, in Subsection \ref{anprelim}, we set out some preparatory results concerning the smoothness of various functions that arise in our arguments. Again, throughout this section, we fix $I\in\mathbb{N}$.

\subsection{Analytic preliminaries}\label{anprelim}

We begin by preparing some notation. For any function $f : \R_+ \to \R$, we define extended functions $f^{\R} :\R \to \R$ and $f^{\R, \uparrow}: \R \to \R$ by setting
\[f^{\R}(u):=\left\{\begin{array}{ll}
               f(u), & u\geq 0,\\
               0, & u<0,
             \end{array}\right.\]
and
\[f^{\R,\uparrow}(u):=\left\{\begin{array}{ll}
               f(u), & u\geq 0,\\
               u, & u<0.
             \end{array}\right.\]
For $q \in \N$, we then define
\begin{align*}
\mathcal{C}^{(q)}&:=\left\{f\in \mathcal{C} :\: f^{\R} \in C^q(\R,\R) \right\},\\
\mathcal{C}^{\uparrow(q)}&:=\left\{f\in \mathcal{C}^{\uparrow} :\: f^{\R, \uparrow} \in  C^q(\R,\R) \right\}.
\end{align*}
Note that $\mathcal{C}^{(1)} = \mathcal{C} \cap \mathcal{C}^1(\R_+,\R_+)$, where $\mathcal{C}^1(\R_+,\R_+)$ was defined in the introduction. Also, for $f \in \mathcal{C}^{(q)}$, $f'(0)=0$ and for $f \in \mathcal{C}^{\uparrow(q)}$, $f'(0)=1$. We have the following further basic properties for functions in these sets.

\begin{lem}\label{derivative}
Let $q \in \N$.  \\
(a) If $f  \in \mathcal{C}^{\uparrow(q)}$ and $f'(u) >0 $ for all $u \in \R_+$, then $f^{-1} \in  \mathcal{C}^{\uparrow(q)}$.\\
(b) Suppose $f  \in C^q(\R_+^2,\R_+)$, $f(\cdot, t)  \in \mathcal{C}^{\uparrow}$ for all $t \in \R_+$, and $\partial_u f(u,t) >0 $ for all $(u,t) \in \R_+^2$. If $f^{-1}(\cdot,t)$ is defined to be the inverse function of $f(\cdot,t)  :\R_+ \to \R_+$ for each $t \ge 0$, then $f^{-1} \in C^q(\R_+^2,\R_+)$.\\
(c) If $f \in \mathcal{C}^{(q)}$ and $g \in \mathcal{C}^{(q)} \cup \mathcal{C}^{\uparrow(q)}$, then $f \circ g \in \mathcal{C}^{(q)}$.
\end{lem}
\begin{proof}
(a) Since $f^{\R, \uparrow}$ is continuous, strictly increasing, and it also holds that $\lim_{u \to \pm \infty} f^{\R, \uparrow}(u)= \pm \infty$, it is the case that $(f^{\R, \uparrow})^{-1} :\R \to \R$ exists. It is moreover straightforward to check that $(f^{\R, \uparrow})^{-1} = (f^{-1})^{\R,\uparrow}$. As $(f^{\R, \uparrow})'(u) >0$ for all $u \in \R$, the remaining differentiability claim follows from the inverse function theorem.\\
(b) If we define $\tilde{f} : \R_+^2 \to \R_+^2$ by setting $\tilde{f}(u,t):=(f(u,t),t)$, then $\tilde{f}$ is a bijection with inverse given by $\tilde{f}^{-1}(u,t)=(f^{-1}(u,t),t)$. Since the Jacobian determinant of $\tilde{f}$ is equal to $\partial_u f(u,t)$, which is strictly positive by assumption, one can again apply the inverse function theorem to deduce that $\tilde{f}^{-1} \in C^q(\R_+^2,\R_+^2)$, which implies in turn that $f^{-1} \in  C^q(\R_+^2,\R_+)$.\\
(c) Since $f \circ g \in \mathcal{C}$ and $(f \circ g)^{\R}=f^{\R} \circ g^{\R} = f^{\R} \circ g^{\R, \uparrow}$, $f \circ g \in \mathcal{C}^{(q)}$.
\end{proof}

We next consider the incorporation of regularity criteria to the domains $\mathcal{D}$ and $\bar{\mathcal{D}}$, as defined at \eqref{ddef} and \eqref{dbardef}. In particular, we set
\[\mathcal{D}^{(q)}:= \mathcal{D} \cap (\mathcal{C}^{(q)})^I,\qquad \bar{\mathcal{D}}^{(q)}:= \bar{\mathcal{D}} \cap  (\mathcal{C}^{(q)})^I.\]
Note that $\mathcal{D}^{(1)}=\mathcal{D} \cap \mathcal{C}^1(\R_+,\R_+)^I$, which is consistent with the definition given in the statement of Theorem \ref{mainthm2}. We make the following observations.

\begin{lem}\label{regularity}
Let $q \in \N$. \\
(a) For $\psi \in \mathcal{D}^{(q)}$ and $i=1,2,\dots, I$, it holds that $\phi_i \in \mathcal{C}^{\uparrow(q)}$, and $\phi_i' (u) >0$ for all $u \in \R_+$.\\
(b) For $\bar{\psi } \in \bar{\mathcal{D}}^{(q)}$, it holds that $\gamma_i \in \mathcal{C}^{\uparrow(q)}$ for $i=0,1,\dots,I$, and $\gamma_0' (z)>0$ for all $z \in \R_+$.\\
(c) The map $\Upsilon$ is a bijection from $\mathcal{D}^{(q)}$ to $\bar{\mathcal{D}}^{(q)}$.
\end{lem}
\begin{proof}
(a) By Lemma \ref{philem}, $\phi_i \in \mathcal{C}^{\uparrow}$. Moreover, by definition,
\[\phi_i^{\R, \uparrow}(u)= u-\sum_{j=1}^I2(i\wedge j) \psi_j^{\R}(u),\]
which shows $\phi_i^{\R, \uparrow}(u) \in C^q(\R,\R)$, and so we obtain that $\phi_i \in \mathcal{C}^{\uparrow(q)}$. Since
\[\phi_i '(u)= 1- \sum_{j=1}^I 2 (i \wedge j)  \psi_j' (u)\le 1,\]
for any $u\in\mathbb{R}_+$ and $\varepsilon \neq 0$ such that $u+\varepsilon \in \R_+$,
\[\sum_{i=1}^I i \frac{\psi_i  (u+\varepsilon)-\psi_i(u)}{\varepsilon} \le \sum_{i=1}^I i \frac{\psi_i  (u+\varepsilon)-\psi_i(u) }{\phi_i  (u+\varepsilon)-\phi_i(u)}.\]
Taking the limit as $\varepsilon\rightarrow0$ and applying the assumption that $\psi \in \mathcal{D}$, we thus find that $ \sum_{i=1}^I  i  \psi_i ' (u) <\frac{1}{2}$. This implies $\phi_i ' (u) \ge  1- \sum_{j=1}^I 2 j  \psi_j ' (u) >0$.  \\
(b) By Lemma \ref{ttt1}, $\gamma_i \in \mathcal{C}^{\uparrow}$ for $i=0,1,2,\dots, I$. Also by definition, $\gamma_I^{\R,\uparrow}(z)=z$ and
\[\gamma_i^{\R,\uparrow}(z)=z+\sum_{j>i}2\left(j-i\right)\bar{\psi}_j^{\R}\circ \gamma_j^{\R,\uparrow}(z),\qquad i=I-1,I-2,\dots,0.\]
Hence, inductively, we have $\gamma_i^{\R,\uparrow} \in C^q(\R,\R)$ and $\gamma_i \in \mathcal{C}^{\uparrow(q)}$. In particular, since
\[\gamma_0(z)=z+\sum_{i}2i \bar{\psi}_i\circ \gamma_i(z) \]
and $\bar{\psi}_i$ and $\gamma_i$ are non-decreasing, $\gamma_0 ' (z) \ge 1 >0$.  \\
(c) We only need to prove that $\Upsilon(\mathcal{D}^{(q)}) \subseteq \bar{\mathcal{D}}^{(q)}$ and $\Gamma ( \bar{\mathcal{D}}^{(q)}) \subseteq \mathcal{D}^{(q)}$. First suppose $\psi \in \mathcal{D}^{(q)}$. By Lemma \ref{derivative}(a) and part (a) of this one, we find that $\phi_i^{-1} \in \mathcal{C}^{\uparrow(q)}$. Hence, by Lemma \ref{derivative}(c), $\psi_i \circ \phi_i^{-1} \in \mathcal{C}^{(q)}$, and so $\Upsilon (\psi) \in \bar{\mathcal{D}}^{(q)}$. Next, suppose $\bar{\psi} \in \bar{\mathcal{D}}^{(q)}$. By Lemma \ref{derivative}(a) and part (b) of this one, $\gamma_0^{-1} \in  \mathcal{C}^{\uparrow(q)}$. So, by Lemma \ref{derivative}(c), $\bar{\psi}_i \circ \gamma_i \circ \gamma_0^{-1} \in  \mathcal{C}^{(q)}$, which allows us to conclude that $\Gamma(\bar{\psi}) \in  \mathcal{D}^{(q)}$.
\end{proof}

In the following two lemmas, we give dynamical versions of the preceding results. To state these, for $f, g \in  C(\R_+^2, \R_+)$, we define a composed function $f \star g \in C(\R_+^2, \R_+)$ by setting
\[f \star g (u,t):=f(g(u,t),t).\]
Note that $(f \star g) \star h= f \star (g \star h)$.

\begin{lem}\label{derivative2}
Let $q \in \N$. If $f, g \in C^q(\R_+^2, \R_+)$, then $f \star g  \in C^q(\R_+^2, \R_+)$. If it further holds that $g(\cdot, t) \in \mathcal{C}^{\uparrow}$ for all $t \in \R_+$, and $\partial_u g(u,t) >0$ for all $(u,t) \in  \R_+^2$, then $f \star g^{-1} \in C^q(\R_+^2, \R_+)$, where $g^{-1}(\cdot, t)$ is defined to be the inverse function of $g(\cdot,t) :\R_+ \to \R_+$ for each $t \in \R_+$.
\end{lem}
\begin{proof} For $\tilde{g}(u,t):=(g(u,t),t)$, $f \star g= f \circ \tilde{g}$. Since $f \in C^q(\R_+^2, \R_+)$ and $\tilde{g} \in C^q(\R_+^2, \R_+^2)$, the first claim of the lemma follows. For the second claim, we apply Lemma \ref{derivative}(b) to deduce that $g^{-1} \in C^q(\R_+^2, \R_+)$. Hence applying the first part of the lemma completes the proof.
\end{proof}

 \begin{lem}\label{jointregularity}
Let $q \in \N$. Suppose $\psi(\cdot,t)\in \mathcal{F}$ and set $\bar{\psi}(\cdot,t) := \Upsilon (\psi(\cdot,t))$ for all $t \in \R_+$. The following two conditions are equivalent. \\
(i) It holds that $\psi(\cdot,t) \in \mathcal{D}^{(q)}$ for all $t \ge 0$, and $\psi \in C^q(\mathbb{R}_+^2,\mathbb{R}_+)^{I}$.\\
(ii) It holds that $\bar{\psi}(\cdot,t) \in \bar{\mathcal{D}}^{(q)}$ for all $t \ge 0$, and $\bar{\psi}  \in C^q(\mathbb{R}_+^2,\mathbb{R}_+)^{I}$.
\end{lem}
\begin{proof}
By Lemma \ref{regularity} (c), it is clear that the following two conditions are equivalent.\\
(i') It holds that $\psi(\cdot,t) \in \mathcal{D}^{(q)}$ for all $t \ge 0$.\\
(ii') It holds that $\bar{\psi}(\cdot,t) \in \bar{\mathcal{D}}^{(q)}$ for all $t \ge 0$.\\
Hence it remains to relate the joint regularity of $\psi$ and $\bar{\psi}$. By construction, we have that $\bar{\psi}_i= \psi_i \star \phi_i^{-1}$ and $\psi_i = \bar{\psi}_i \star \gamma_i \star \gamma_0^{-1}$, where, for each $t \in \R_+$, $\phi_i^{-1}(\cdot, t)$ is the inverse function of $\phi_i(\cdot,t) :\R_+ \to \R_+$, and $\gamma_0^{-1}(\cdot, t)$ is the inverse function of $\gamma_0(\cdot,t) :\R_+ \to \R_+$. Now suppose (i) holds. By definition, $\phi_i \in C^q(\mathbb{R}_+^2,\mathbb{R}_+)$ for $i=1,2,\dots, I$. Moreover, by Lemma \ref{regularity}(a), $\phi_i(\cdot, t) \in \mathcal{C}^{\uparrow}$ for all $t \in \R_+$, and $\partial_u \phi_i(u,t)>0$ for all $(u,t) \in \R_+^2$. Hence, by Lemma \ref{derivative2}, $ \bar{\psi}_i= \psi_i \star \phi_i^{-1} \in C^q(\R_+^2, \R_+)$ for $i=1,2,\dots, I$, and so (ii) holds. Next, suppose (ii) holds. Arguing similarly to the previous case, we have that $\bar{\psi}_i, \gamma_i, \gamma_0^{-1} \in C^q(\R_+^2, \R_+)$, and thus Lemma \ref{derivative2} yields that $( \bar{\psi}_i \star \gamma_i) \star \gamma_0^{-1} \in C^q(\R_+^2, \R_+)$. This confirms that (i) also holds.
\end{proof}

Next, we connect the density and regularity conditions for $\rho$ and $\psi$. Recall the definition of $\mathcal{D}_{\mathrm{density}}$ from above Theorem \ref{reallythemainthm}.

\begin{lem}\label{rhopsirelation}
Suppose $\rho=(\rho_i) \in C(\R_+,\R_+)^I$. The following two conditions are equivalent. \\
(i) It holds that $\rho \in \mathcal{D}_{\mathrm{density}}$. \\
(ii) It holds that $\psi \in \mathcal{D}^{(2)}$, where $\psi_i(u):=\int_0^u \rho_i(u')du'$.
\end{lem}
\begin{proof}
Suppose (i) holds. It then clearly holds that $\psi \in \mathcal{C}^I$. Moreover, since $\psi_i^{\R}(u)=\int_{-\infty}^u \rho_i^{\R}(u')du'$, and $\rho_i^{\R} \in C^1(\R,\R)$ by assumption, $\psi_i^{\R} \in C^2(\R,\R)$. Hence $\psi_i \in \mathcal{C}^{(2)}$. Next, since $\rho$ satisfies \eqref{cond1} and \eqref{cond2},
\[\sup_{u\in\mathbb{R}_+}\sum_{i=1}^I2i\rho_i(u)\leq \sum_{i=1}^I i \sup_{u\in\mathbb{R}_+}\frac{\rho_i(u)}{1-\sum_{j=1}^I2(i\wedge j)\rho_j(u)}<\frac{1}{2},\]
and so $\inf_{u\in\R_+}\phi_I'(u)=\inf_{u\in\R_+}(1-2\sum_{i}i\rho_i (u)) >0$. Hence $\phi_I\in \mathcal{C}^{\uparrow}$. Now we study the density condition. First, as in the proof of Lemma \ref{regularity}(a), we see that $\phi_i \in \mathcal{C}^{\uparrow(2)}$ for all $i$. Moreover, since $\phi_i'(u) \ge \phi_I'(u)>0$ for all $u \in \R_+$, applying Lemma \ref{derivative}(a,c) yields that $\phi_i^{-1} \in \mathcal{C}^{\uparrow(2)}$ and $\bar{\psi}_i := \psi_i \circ \phi_i^{-1} \in \mathcal{C}^{(2)}$. By making a change of variables as in the proof of Lemma \ref{tpprop}(a), we have that
\begin{align*}
\sup_{u_1,u_2}\frac{\psi_i(u_1)-\psi_i(u_2)}{\phi_i(u_1)-\phi_i(u_2)}=\sup_{z_1,z_2}\frac{\bar{\psi}_i(z_1)-\bar{\psi}_i(z_2)}{z_1-z_2},
\end{align*}
and since $\psi_i'(u)=\bar{\psi}_i'(\phi_i(u))\phi_i'(u)$, this expression can be rewritten as follows:
\[\sup_{z_1,z_2}\frac{\bar{\psi}_i(z_1)-\bar{\psi}_i(z_2)}{z_1-z_2} = \sup_{z}\bar{\psi}_i'(z)= \sup_{u} \frac{\psi_i'(u)}{\phi_i'(u)}=\sup_u \frac{\rho_i(u)}{1-\sum_{j}2(i\wedge j)\rho_j(u)}.\]
Therefore, appealing to \eqref{cond2}, we can conclude that $\psi \in \mathcal{D}^{(2)}$.

Next, suppose (ii) holds. Since $\psi^{\R}_i \in C^2(\R,\R)$ implies $\rho_i \in C^1(\R_+,\R_+)$, $\rho_i(0)=0$ and $\rho_i'(0)=0$, we find that $\rho \in \mathcal{C}^1(\R_+,\R_+)^I$. Also, by Lemma \ref{regularity}, $\phi_i \in \mathcal{C}^{\uparrow(2)}$ and $\phi_i' (u)>0$ for all $u\in\R_+$, and so  we have $1- \sum_{i}2i\rho_i(u)= \phi_I'(u) >0 $ for all $u \in \R_+$, which gives \eqref{cond1}. Moreover, as we have
\[\sup_u \frac{\rho_i(u)}{1-\sum_{j}2(i\wedge j)\rho_j(u)} = \sup_{u} \frac{\psi_i'(u)}{\phi_i'(u)} \le \sup_{u_1,u_2}  \frac{\psi_i(u_1)-\psi(u_2)}{\phi_i(u_1)-\phi_i(u_2)},\]
the condition at \eqref{cond2} also holds.
\end{proof}

As the last of the preliminaries, we give a regularity property and a partial differential equation for dynamics given by the simple shift operator defined at \eqref{thetatdef}. We recall that $v_i:=i$.

\begin{lem}\label{barpsidynamics}
Let $q \in \N$. Suppose $\bar{\psi}^0 \in \bar{\mathcal{D}}^{(q)}$, and define $\bar{\psi}(z,t):=(\theta_t \bar{\psi}^0)(z)$. It is then the case that $\bar{\psi} \in C^q(\R_+^2,\R_+)^I$ and
\begin{equation}\label{barpsiPDEini}
\begin{cases}
&\vspace{3pt} \displaystyle{\partial_t \bar{\psi}_i =- v_i \partial_z \bar{\psi}_i}, \qquad i=1,2,\dots,I,  \\
& \displaystyle{\bar{\psi}(\cdot,0)=\bar{\psi}^0.}
\end{cases}
\end{equation}
Moreover, $\bar{\psi}(\cdot,t) \in \bar{\mathcal{D}}^{(q)}$ for all $t \in \R_+$.
\end{lem}
\begin{proof}
By definition $\bar{\psi}_i(z,t)=\bar{\psi}_i^0((z-v_it) \vee 0)=(\psi^0_i)^{\R}(z-v_it)$. Since $(\psi^0_i)^{\R} \in C^q(\R,\R)$ by assumption, the function $\bar{\psi}_i^{\R^2}(z,t):=(\psi^0_i)^{\R}(z-v_it)$ on $\R^2$ is in $C^q(\R^2,\R)$, and satisfies
\[\partial_t\bar{\psi}_i^{\R^2} =- v_i \partial_z \bar{\psi}_i^{\R^2}\]
on $\R^2$. Since $\bar{\psi}_i(z,t)$ is the restriction of $\bar{\psi}_i^{\R^2}$ to $\R_+^2$, we thus have $\bar{\psi} \in C^q(\R_+^2,\R_+)^I$ and also \eqref{barpsiPDEini} holds. Finally, since $(\bar{\psi}_i(\cdot,t)^{\R})(z)=\bar{\psi}_i^{\R^2}(z,t)$ for all $t \in \R_+$, $\bar{\psi}_i(\cdot,t) \in \mathcal{C}^{(q)}$. In conjunction with Proposition \ref{tpprop}(b), this completes the proof.
\end{proof}

\subsection{Derivation of partial differential equations}\label{pdesubsec}

In this subsection, we put together the various results we have established so far to complete the proofs of Theorems \ref{reallythemainthm} and \ref{mainthm2}, as well as that of Corollary \ref{maincor1}. To this end, we start by giving a lemma that connects the free and effective speeds of solitons in a slightly different way to \eqref{effv}. In particular, we now define $\kappa_{ij}:=2( i \wedge j)$, and for $\rho\in\mathbb{R}_+^I$, introduce a matrix $M^*=M^*(\rho)$ by setting
\begin{equation}\label{m*matrix}
\left\{
  \begin{array}{ll}
    \displaystyle{M_{ii}^*(\rho)=1-\sum_{j \neq i}\kappa_{ij}\rho_j}, &\vspace{3pt} \\
    \displaystyle{M_{ij}^*(\rho)=\kappa_{ij}\rho_i}, & i\neq j.
  \end{array}
\right.
\end{equation}
Observe that, since $\kappa_{ij}=\kappa_{ji}$, $M^*$ is the transpose of $M=M(\rho)$, as defined at \eqref{matrix}. (The reason we define $M^*$ in the way we do is to highlight that the argument we give in the next two results applies more generally than to the symmetric case of interest in the present work, see Remark \ref{follrem} below.) Again recall that $v_i=i$.

\begin{lem}\label{M-matrix} If $\psi \in \mathcal{D}^{(1)}$, then
\[v_i \psi_i'=\sum_{j=1}^{I}M_{ij}^*(\psi')v^{\mathrm{eff}}_j(\psi') \psi_j'. \]
\end{lem}
\begin{proof}
By the definition of $v^{\mathrm{eff}}$ (i.e.\ \eqref{effv}), we have that
\[v_i=\sum_{j=1}^I M_{ij}(\psi') v_j^{\mathrm{eff}}(\psi').\]
Multiplying both sides by $\psi_i'$, we obtain
\begin{align*}
v_i \psi_i' & = \sum_{j=1}^{I} M_{ij}(\psi') v_j^{\mathrm{eff}}(\psi')\psi_i' \\
&=  \left(1- \sum_{j \neq i} \kappa_{ij} \psi_j' \right) v_i^{\mathrm{eff}}(\psi') \psi_i' +  \sum_{j \neq i}  \kappa_{ij} \psi_j'   v_j^{\mathrm{eff}}(\psi') \psi_i'\\
 & = M_{ii}^* (\psi') v_i^{\mathrm{eff}}(\psi') \psi_i'  +  \sum_{j \neq i}M_{ij}^* (\psi') v_j^{\mathrm{eff}}(\psi') \psi_j',
 \end{align*}
 which completes the proof.
 \end{proof}

The next lemma relates the partial differential equation for the integrated densities of solitons in space with the corresponding equation for solitons on their effective scale.

\begin{lem}\label{equivalenceofPDE}
Let $q \in \N$. Suppose $\psi(\cdot,t) \in \mathcal{F}$ and set $\bar{\psi}(\cdot,t) := \Upsilon (\psi(\cdot,t))$ for all $t \in \R_+$. The following two conditions are equivalent.  \\
(i) It holds that $\psi(\cdot,t) \in \mathcal{D}^{(q)}$ for all $t \in \R_+$, $\psi \in C^q(\mathbb{R}_+^2,\mathbb{R}_+)^{I}$ and
\begin{equation}\label{psiPDE2}
\partial_t \psi_i =- v_i^{\mathrm{eff}}(\partial_u \psi) \partial_u \psi_i.
\end{equation}
(ii) It holds that $\bar{\psi} (\cdot,t) \in \bar{\mathcal{D}}^{(q)}$ for all $t \in \R_+$, $\bar{\psi}  \in C^q(\mathbb{R}_+^2,\mathbb{R}_+)^{I}$ and
\begin{equation}\label{barpsiPDE}
\partial_t \bar{\psi}_i =- v_i \partial_z \bar{\psi}_i.
\end{equation}
\end{lem}
\begin{proof} First note that under either condition (i) or (ii), Lemma \ref{jointregularity} yields that $\psi(\cdot,t) \in \mathcal{D}^{(q)}$ for all $t \in \R_+$, $\psi \in C^q(\mathbb{R}_+^2,\mathbb{R}_+)^{I}$, $\bar{\psi}(\cdot,t) \in \bar{\mathcal{D}}^{(q)}$ for all $t \in \R_+$ and $\bar{\psi}  \in C^q(\mathbb{R}_+^2,\mathbb{R}_+)^{I}$. Hence, we only need to prove that under these four conditions, \eqref{psiPDE2} and \eqref{barpsiPDE} are equivalent. This being the case, we will assume that these four conditions hold for the remainder of the proof.

By Lemma \ref{M-matrix}, we have that
\begin{equation}\label{relation1}
v_i \partial_u \psi_i = \sum_{j=1}^{I}M_{ij}^* (\partial_u \psi) v_j^{\mathrm{eff}}(\partial_u \psi) \partial_u \psi_j.
\end{equation}
Also, by direct computation,
\begin{align*}
\sum_{j=1}^IM_{ij}^*(\partial_u \psi) \partial_t \psi_j & =  \left(1 -\sum_{j \neq i} \kappa_{ij}  \partial_u  \psi_j\right)   \partial_t \psi_i  + \sum_{ j \neq i} \kappa_{ij}  \partial_u  \psi_i  \partial_t \psi_j \\
& =  \left(1 -\sum_{j=1 }^I \kappa_{ij}  \partial_u  \psi_j  \right) \partial_t \psi_i  + \sum_{j=1 }^I \kappa_{ij}  \partial_u  \psi_i  \partial_t \psi_j.
\end{align*}
Substituting $\partial_u \phi_i= 1- \sum_{j=1}^I \kappa_{ij}\partial_u\psi_j$ and $\partial_t \phi_i= - \sum_{j=1}^I \kappa_{ij}\partial_t\psi_j$ into the above relation, we obtain
\begin{equation}\label{relation2}
\sum_{j=1}^IM_{ij}^*(\partial_u \psi) \partial_t \psi_j = \partial_t \psi_i \partial_u \phi_i - \partial_u  \psi_i   \partial_t \phi_i.
\end{equation}
Furthermore, since $\psi_i(u,t)=\bar{\psi}_i(\phi_i(u,t),t)$, applying the chain rule yields
\[\partial_u \psi_i (u,t)= (\partial_z\bar{\psi}_i) ( \phi_i(u,t),t) (\partial_u \phi_i) (u,t)\]
and also
\[\partial_t \psi_i (u,t)= (\partial_t \bar{\psi}_i)  ( \phi_i(u,t),t) + (\partial_z\bar{\psi}_i) ( \phi_i(u,t),t) (\partial_t \phi_i) (u,t).\]
In particular, the latter two equalities are rearranged to give
\begin{equation}\label{derivativerelation1}
 (\partial_z\bar{\psi}_i) ( \phi_i(u,t),t) = \frac{ \partial_u \psi_i (u,t)}{\partial_u \phi_i (u,t)},
 \end{equation}
 and
\begin{equation}\label{derivativerelation2}
(\partial_t \bar{\psi}_i)  ( \phi_i(u,t),t)  = \partial_t \psi_i (u,t)- \frac{ \partial_u \psi_i (u,t)}{\partial_u \phi_i (u,t)} \partial_t \phi_i (u,t),
\end{equation}
where we remind the reader that, by Lemma \ref{regularity}, $\partial_u \phi_i (u,t) >0$ for all $u,t \in \R_+^2$ under the four conditions that we are assuming. Combining equations \eqref{relation1}, \eqref{relation2}, \eqref{derivativerelation1} and \eqref{derivativerelation2}, we thus find that
\begin{align*}
v_i (\partial_z\bar{\psi}_i) ( \phi_i(u,t),t) &= \frac{1}{\partial_u \phi_i } \sum_{j=1}^{I}M_{ij}^* (\partial_u \psi) v_j^{\mathrm{eff}}(\partial_u \psi) \partial_u \psi_j, \\
(\partial_t \bar{\psi}_i)  ( \phi_i(u,t),t)  & =   \frac{1}{\partial_u \phi_i} \sum_{j=1}^IM_{ij}^*(\partial_u \psi) \partial_t \psi_j,
\end{align*}
where we omit the variables $(u,t)$ in the right hand side.

Now, suppose \eqref{psiPDE2} holds. It then follows from the above relations that
\[(\partial_t \bar{\psi}_i)  ( \phi_i(u,t),t) =- v_i (\partial_z\bar{\psi}_i) ( \phi_i(u,t),t)\]
for all $(u,t)$, which implies \eqref{barpsiPDE}. On the other hand, if we assume \eqref{barpsiPDE}, then we have
\[\sum_{j=1}^IM_{ij}^*(\partial_u \psi) \partial_t \psi_j = -\sum_{j=1}^{I}M_{ij}^* (\partial_u \psi) v_j^{\mathrm{eff}}(\partial_u \psi) \partial_u \psi_j\]
for all $i$. Since $M(\partial_u \psi)$ is invertible by Lemma \ref{minvert}, its transpose $M^*(\partial_u \psi)$ is also invertible, and so \eqref{psiPDE2} also holds.
\end{proof}

\begin{rem}\label{follrem}
To show the equivalence of the two partial differential equations \eqref{psiPDE2} and \eqref{barpsiPDE} in the above proof, except the bijectivity between classes of functions with certain regularity, we only use the following relations:
\begin{align*}
 \phi_i(u) & =u-\sum_{j} \kappa_{ij} \psi_j(u), \\
\bar{\psi}_i  & =\psi_i \circ \phi_i^{-1}, \\
 v&=M(\psi')v^{\mathrm{eff}}(\psi'),
\end{align*}
with $M_{ii}(\psi'):=1-\sum_{i \neq j} \kappa_{ij} \psi'_j$ and $M_{ij} :=\kappa_{ij}\psi'_j$ for $i\neq j$, and the invertibility of the matrix $M^*$ defined by \eqref{m*matrix}. Note that we do not use the symmetry of $\kappa$. Hence, for more general integrable systems with isolated soliton speeds $v=(v_i)$ and phase shifts $\kappa=(\kappa_{ij})$, once we can establish the bijectivity of the map $\psi \mapsto \bar{\psi}$ given by the above relation in a class of functions with appropriate regularity, and the invertibility of $M^*$, we can apply the same argument.
\end{rem}

\begin{proof}[Proof of Theorem \ref{mainthm2}]
(a) Under the assumption, Lemma \ref{regularity} gives us that $\bar{\psi}^0:=\Upsilon(\psi^0) \in \bar{\mathcal{D}}^{(1)}$. If we define $\bar{\psi}_i(z,t):=(\theta_t \psi^0)_i (z)$, then Lemma \ref{barpsidynamics} implies $\bar{\psi} (\cdot,t) \in \bar{\mathcal{D}}^{(1)}$ for all $t \in \R_+$, $\bar{\psi}  \in C^1(\mathbb{R}_+^2,\mathbb{R}_+)^{I}$ and \eqref{barpsiPDE} holds. Since $\bar{\psi} (\cdot,t) =\Upsilon( \psi(\cdot,t))$, from Lemma \ref{equivalenceofPDE}, we can conclude that $\psi \in C^1(\mathbb{R}_+^2,\mathbb{R}_+)^{I}$, $\psi (\cdot,t) \in \mathcal{D}^{(1)}$ for all $t \in \R_+$ and \eqref{PDEpsi} holds. It remains to verify the uniqueness claim. Suppose there are two solutions $\psi$ and $\tilde{\psi}$ to \eqref{PDEpsi} satisfying $\psi,\tilde{\psi} \in C^1(\mathbb{R}_+^2,\mathbb{R}_+)^{I}$, $\psi (\cdot,t), \tilde{\psi}(\cdot,t) \in \mathcal{D}^{(1)}$ for all $t \in \R_+$. By Lemma \ref{equivalenceofPDE}, it is then the case that $\bar{\psi} (\cdot,t) :=\Upsilon( \psi(\cdot,t))$ and $\bar{\tilde{\psi}} (\cdot,t) =\Upsilon( \tilde{\psi}(\cdot,t))$ satisfy the partial differential equation \eqref{barpsiPDE} with the same initial condition $\Upsilon(\psi^0)$. Since the classical solution of \eqref{barpsiPDE} is clearly unique, $\bar{\psi} =\bar{\tilde{\psi}}$, and so $\psi=\tilde{\psi}$.\\
(b) By Lemma \ref{rhopsirelation}, we have $\psi^0 \in \mathcal{D}^{(2)}$. By the same argument as was used to establish (a), it follows that $\psi \in C^2(\mathbb{R}_+^2,\mathbb{R}_+)^{I}$, $\psi$ satisfies $\psi (\cdot,t) \in \mathcal{D}^{(2)}$ for all $t \in \R_+$, and \eqref{PDEpsi} holds. Therefore, $\rho \in  C^1(\R_+^2, \R)^I$ and by Lemma \ref{rhopsirelation}, $\rho(\cdot, t) \in \mathcal{D}_{\mathrm{density}}$ for all $t \in \R_+$. Also, by taking space derivatives of both sides of \eqref{PDEpsi}, $\rho$ satisfies \eqref{PDErho}. Finally, regarding uniqueness, suppose there are two solutions $\rho,\tilde{\rho}$ of \eqref{PDErho} satisfying $\rho(\cdot, t) \in \mathcal{D}_{\mathrm{density}}$ for all $t \in \R_+$ and $\rho, \tilde{\rho} \in C^1(\R_+^2,\R_+)$. Defining $\psi_i(u,t):=\int_0^u \rho_i(u',t)du'$ and $\tilde{\psi}_i(u,t):=\int_0^u \tilde{\rho}_i(u',t)du'$, we obtain that $\psi,\tilde{\psi} \in C^1(\mathbb{R}_+^2,\mathbb{R}_+)^{I}$ and, from Lemma \ref{rhopsirelation}, $\psi (\cdot,t), \tilde{\psi}(\cdot,t) \in \mathcal{D}^{(2)}$ for all $t \in \R_+$. Moreover, integrating \eqref{PDErho} with respect to $u$ gives \eqref{PDEpsi}, and so we can apply the uniqueness result for $\psi, \tilde{\psi}$ to complete the proof.
\end{proof}

\begin{lem}\label{convergencequivalence}
For each $N\in\mathbb{N}$, suppose $\eta$ is a random configuration with distribution $\mathbf{P}_N$ supported on $\Omega_I$ for some $I\in \mathbb{N}$ (independent of $N$). For $\rho=(\rho_i)_{i=1}^I\in \mathcal{D}_{\mathrm{density}}$, the following conditions are equivalent. \\
(i) For every $\varepsilon>0$ and $(F_i)_{i=1}^I \in C_0(\mathbb{R}_+,\mathbb{R})^I$,
\[\lim_{N\rightarrow\infty} \mathbf{P}_N\left(\sup_{i\in\{1,2,\dots,I\}} \left| \int_{\R_+} F_i(u) \pi_i^{N,0} (du)  - \int_{\R_+} F_i(u) \rho_i(u)du  \right|  > \varepsilon\right) =0.\]
(ii) For every $\varepsilon>0$ and $u_0\in(0,\infty)$,
\[\lim_{N\rightarrow\infty}\mathbf{P}_N\left(\sup_{i\in\{1,2,\dots,I\}}\sup_{u\leq u_0}\left|\psi_i^N(u,0)-\psi_i(u)\right|>\varepsilon\right)=0,\]
where $\psi_i(u)=\int_0^u \rho_i(u')du'$.
\end{lem}
\begin{proof} Suppose (i) holds. For each $u\in \R_+$, by approximating the function $\mathbf{1}_{[0,u]}$ uniformly by functions in $C_0(\mathbb{R}_+,\mathbb{R})$, one can then deduce that
\begin{equation}\label{pointconvergence}
\lim_{N\rightarrow\infty}\mathbf{P}_N\left(\sup_{i\in\{1,2,\dots,I\}} \left|\psi_i^N(u,0)-\psi_i(u)\right|>\varepsilon\right)=0.
\end{equation}
Since each of the functions  $\psi_i^N(\cdot, 0)$ is monotonic, and each $\psi_i$ is continuous, one may then extend the above to a uniform convergence statement. Indeed, for any  $u_0,\delta\in(0,\infty)$, there exists a finite sequence $0=u_{(0)} < u_{(1)} < u_{(2)} < \dots < u_{(\ell)} = u_0$ such that
\[\sup_{i\in\{1,2,\dots,I\}}\sup_{n \in \{0,1,2,\dots,\ell-1\}} \left( \psi_i(u_{(n+1)})- \psi_i (u_{(n)})\right) < \delta.\]
Moreover, appealing to monotonicity, we have
\begin{align*}
 \sup_{u \leq u_0} \left|\psi_i^N(u,0)-\psi_i(u)\right|  &=\sup_{n \in \{0,1,2,\dots,\ell-1\}}  \sup_{u \in [u_{(n)}, u_{(n+1)}]} \left|\psi_i^N(u,0)-\psi_i(u)\right|  \\
& \le \sup_{n \in \{0,1,2,\dots,\ell\}}\left| \psi_i^N(u_{(n)},0)- \psi_i (u_{(n)})\right| \\
&\hspace{40pt}+  \sup_{n \in \{0,1,2,\dots,\ell-1\}} \left( \psi_i(u_{(n+1)})- \psi_i (u_{(n)})\right).
\end{align*}
Combining these inequalities with \eqref{pointconvergence} yields (ii).

Next, suppose (ii) holds. One can then use the fact that any $F \in  C_0(\mathbb{R}_+,\mathbb{R})$ is approximated uniformly by a simple function to check that (i) holds.
\end{proof}

\begin{proof}[Proof of Theorem \ref{reallythemainthm}]
Under the assumption, by Lemma \ref{convergencequivalence}, \eqref{initconv} holds with $\psi^0_i(u):=\int^u_0 \rho^0_i(u)du'$. Note in particular that $\psi^0\in\mathcal{D}^{(2)}$. Thus, by Theorem \ref{mainthm} and Lemma \ref{rhopsirelation}, for every $\varepsilon>0$, $t\in(0,\infty)$ and $u_0\in(0,\infty)$,
\[\lim_{N\rightarrow\infty}\mathbf{P}_N\left(\sup_{i\in\{1,2,\dots,I\}}\sup_{u\leq u_0}\left|\psi_i^N(u,t)-\psi_i(u,t)\right|>\varepsilon\right)=0,\]
where $\psi(u,t)$ is defined by \eqref{BBSflow}.
Now, by Theorem \ref{mainthm2}, we have that $\rho_i(u,t):=\partial_u\psi(u,t)$ satisfies $\rho(\cdot,t) \in \mathcal{D}_{\mathrm{density}}$ for all $t \in \R_+$. Hence applying Lemma \ref{convergencequivalence} again yields that, for every $\varepsilon>0$, $t\in(0,\infty)$ and $(F_i)_{i=1}^I \in C_0(\mathbb{R}_+,\mathbb{R})^I$,
\[\lim_{N\rightarrow\infty} \mathbf{P}_N\left(\sup_{i\in\{1,2,\dots,I\}} \left| \int_{\R_+} F_i(u)\pi_i^{N,t}(du) - \int_{\R_+} F_i(u) \rho_i(u,t)du  \right|  > \varepsilon\right) =0.\]
We moreover have from Theorem \ref{mainthm2} that $\rho$ is the unique solution of \eqref{PDErho} satisfying $\rho \in C^1(\R_+^2,\R_+)^I$ and $\rho(\cdot,t) \in \mathcal{D}_{\mathrm{density}}$ for all $t \in \R_+$.
\end{proof}

\begin{proof}[Proof of Corollary \ref{maincor1}] Since
\[\left|\frac{1}{N}\sum_{x=1}^{\lfloor Nu\rfloor }\eta\left(x,\lfloor Nt\rfloor \right)-\sum_{i=1}^{I}i\psi_i^N(u,t)\right|\leq\frac{I^2}{N},\]
the first claim readily follows from Theorem \ref{mainthm}. It similarly holds that, for any $F \in C_0(\mathbb{R}_+,\mathbb{R})$,
\[\left| \frac{1}{N} \sum_{x \in \N} F \left(\frac{x}{N}\right)\eta(x, \lfloor Nt \rfloor) - \sum_{i=1}^I i\int_{\R_+} F(u)\pi_i^{N,t} (du)\right|\leq\frac{I^2\sup_u|F(u)|}{N},\]
and so the second claim is a simple consequence of Theorem \ref{reallythemainthm}.
\end{proof}

\section*{Acknowledgements}

The research of both DC and MS was supported by JSPS Grant-in-Aid for Scientific Research (B), 19H01792.

\providecommand{\bysame}{\leavevmode\hbox to3em{\hrulefill}\thinspace}
\providecommand{\MR}{\relax\ifhmode\unskip\space\fi MR }
\providecommand{\MRhref}[2]{%
  \href{http://www.ams.org/mathscinet-getitem?mr=#1}{#2}
}
\providecommand{\href}[2]{#2}

\end{document}